\documentclass{article}
\usepackage{graphicx}
\usepackage{amsfonts}
\usepackage{amsmath}
\usepackage{amssymb}
\usepackage{url}

\usepackage{dsfont}
\usepackage{indentfirst}
\usepackage{enumerate}
\usepackage{booktabs}
\usepackage{multirow}
\usepackage{float}
\usepackage[section]{placeins}

 \usepackage[colorlinks=true,citecolor=blue,hyperfootnotes=false]{hyperref}
\usepackage{amsthm}
\usepackage{color}
\usepackage{tikz}

\definecolor{DarkGreen}{rgb}{0.2,0.6,0.2}

\definecolor{purple}{rgb}{0.6,0.3,0.8}

\usepackage{natbib}

\def\d{\mathrm{d}}
\def\laweq{\buildrel \mathrm{d} \over =}

\newcommand{\var}{\mathrm{Var}}

\newcommand{\E}{\mathbb{E}} 
\newcommand{\p}{\mathbb{P}} 
\newcommand{\R}{\mathbb{R}}
\newcommand{\N}{\mathbb{N}}

\theoremstyle{plain}
\newtheorem{theorem}{Theorem}

\usepackage{aliascnt}
\newtheorem{lemma}{Lemma}
\newtheorem{proposition}{Proposition}
\theoremstyle{definition}

\newtheorem{example}{Example}

\usepackage[misc]{ifsym}

\theoremstyle{remark}
\newtheorem{remark}{Remark}

\newcommand{\id}{\mathds{1}}

\newcommand{\esssup}{\mathrm{ess\mbox{-}sup}}
\newcommand{\essinf}{\mathrm{ess\mbox{-}inf}}

\usepackage[onehalfspacing]{setspace}

\newcommand{\VaR}{\mathrm{VaR}}
\newcommand{\MMD}{\mathrm{MMD}}

\newcommand{\ES}{\mathrm{ES}}

\newcommand{\bz}{\mathbf{z}}
\newcommand{\bx}{\mathbf{x}}
\newcommand{\AES}{\mathrm{AES}}

\usepackage{footmisc}
\title{Submodular Risk Measures}
\author{Ruodu Wang\thanks{Department of Statistics and Actuarial Science, University of Waterloo,  Canada. \Letter~{\scriptsize\url{wang@uwaterloo.ca}}}
\and Jingcheng Yu\thanks{Department of Statistics and Actuarial Science, University of Waterloo,  Canada. \Letter~{\scriptsize\url{j563yu@uwaterloo.ca}}}
}
\date{\today}

\begin{document}
	\maketitle
	\begin{abstract}
    We study submodularity for law-invariant functionals, with particular attention to convex risk measures. Expected losses are modular, and certainty equivalents are submodular exactly when the loss function is convex. Law-invariant coherent risk measures are submodular exactly when they are coherent distortion risk measures, including Expected Shortfall (ES), and several deviation measures are also submodular. Beyond positive homogeneity, submodularity is restrictive for convex risk measures. We give a complete characterization for shortfall risk measures via the Arrow--Pratt measure of risk aversion, show that optimized certainty equivalents are always submodular, and prove that adjusted Expected Shortfall (AES) is submodular only when it reduces to ES. An empirical illustration for daily US equity returns finds no ES submodularity violations, many Value-at-Risk (VaR) violations, and relatively few AES violations.

    \end{abstract}

\section{Introduction}

Submodularity is an important property in mathematics, optimization, economics, and data science, with particular applications in modeling transport costs, dependence, cooperative games, and social production functions.
Without being exhaustive, we refer the reader to \citet{MM04} for submodularity (as well as its counterpart, supermodularity) in decision theory, to \citet{T78} in the context of cost optimization, to \citet{R13} in the context of quantitative risk management, and to \citet{B22} in the context of machine learning and artificial intelligence.

For a lattice $\mathcal{L}$ equipped with the maximum operator $\vee$ and the minimum operator $\wedge$, submodularity of a function $f:\mathcal{L}\to\R$ means
$$
f(X\vee Y)+f(X\wedge Y)\le f(X)+f(Y) \mbox{~~~for all $X,Y\in \mathcal{L}$.}
$$
A common appearance of submodularity is in the context of capacities. A capacity on a $\sigma$-algebra $\mathcal F$ is an increasing function $v:\mathcal F\to\R$ with $v(\varnothing)=0$, and
it is submodular if $v(A\cup B) +v(A\cap B) \le v(A) +v(B)$ for all $A,B\in \mathcal F$, which corresponds to $(\mathcal{L},\vee,\wedge)=(\mathcal F, \cup,\cap)$.

As standard in the literature of risk measures, we study functions on $L^\infty$ with the lattice structure given by the pointwise maximum and minimum.\footnote{This is different from submodularity formulated on sets of portfolio indices, studied by \citet{GG17}; see Appendix \ref{sec:6} for a comparison.} Although submodularity is popular in many fields, its relevance for risk measures has received limited attention. There is, however, a well-known connection: a comonotonic-additive and coherent risk measure must be a Choquet integral with respect to a submodular capacity; see Theorem 4.94 of \citet{FS16}.

A major advance was made by \citet{CC18}, who obtained a remarkable characterization of submodular functions under positive homogeneity and cash invariance (also known as constant additivity).
Translating this result into the language of risk measures, it says that submodular, positively homogeneous, and monetary risk measures must be coherent and comonotonic-additive.
This includes the Expected Shortfall (ES), the standard risk measure for market risk in the banking industry; see, for example, \citet{MFE15}.

The main purpose of this paper is a systematic study of submodularity in the context of risk measures, going beyond the framework of \citet{CC18}.
In Section \ref{sec:3}, we discuss the economic interpretation of submodularity for risk measures; we argue that it represents a notion of reward for ``flattening the risk", which is quite different from diversification.  
We show in Theorem \ref{th:contrast} that submodularity and diversification (represented by quasi-convexity) enjoy two distinct forms of invariance under increasing transformations.

To characterize submodularity, we start in Section \ref{sec:4}  with the expected losses and their certainty equivalents.  It turns out that expected losses are the only law-invariant functionals that are modular. 
Certainty equivalents are submodular if and only if the corresponding loss function is convex. 
In Section \ref{sec:4prime}, we consider distortion risk measures and two general classes of deviation measures. Directly following from the results of \citet{CC18}, a law-invariant coherent risk measure is submodular if and only if it is a coherent distortion risk measure.
This class includes the ES but not the Value-at-Risk (VaR). 
We further find that many classical deviation measures, including the variance, the mean-median deviation, and the Gini deviation, are submodular, whereas the standard deviation is not; see, for example, \citet{RUZ06} and \citet{BW22}. 

\citet{MM08} showed that for norm-continuous cash-invariant  risk measures, submodularity implies convexity.
Since commonly used risk measures are cash invariant, we will mostly focus on risk measures that are convex.
In Section \ref{sec:5}, 
we study submodularity in four classes of convex monetary risk measures that have explicit formulas: the shortfall risk measures of \citet{FS02}, the optimized certainty equivalents (OCE) of \citet{BT07}, adjusted Expected Shortfall (AES) of \citet{BMW22}, and the monotone mean-deviation measures of \citet{HWW26}.
It turns out that, except for the case of OCE, which is simpler, the other three classes all require quite sophisticated analysis to characterize submodularity; thus Section \ref{sec:5} contains the most technical results in this paper.
For shortfall risk measures, we obtain a complete characterization: assuming twice differentiability, submodularity can be characterized via the Arrow--Pratt measure of risk aversion.
AES can only be submodular when it equals a standard ES. For monotone mean-deviation risk measures generated by a distortion risk measure, submodularity can hold only when it is a coherent distortion risk measure.

Section \ref{sec:7} illustrates how rolling historical estimators of VaR, ES, and a two-level AES specification behave under the submodularity test in US equity data.
Section \ref{sec:8} concludes the paper.

\section{Choquet integrals and risk measures}
\label{sec:2}
We work with an atomless probability space $(\Omega,\mathcal F,\p)$, and $L^\infty$ is the set of all bounded random variables on $(\Omega,\mathcal F,\p)$, with $L^\infty$-norm $\|X\|_\infty=\inf\{x\in \R: \p(X>x)=0\}$. 
We treat almost surely equal random variables as identical. 
Elements of $L^\infty$ are interpreted as random financial losses.
Two random variables $X$ and $Y$ are called \emph{comonotonic} if $(X(\omega)-X(\omega'))(Y(\omega)-Y(\omega'))\ge 0$ for all $\omega,\omega'\in \Omega$ ($\p\times\p$ almost surely). We write $X\laweq Y$ if $X$ and $Y$ are equally distributed   under $\p$.

A risk measure is a mapping  $\rho:L^\infty\to\R$, and it may satisfy the following commonly studied properties.
\begin{enumerate}[(a)]
\item Monotonicity: $\rho(X)\le \rho(Y)$ for $X,Y\in L^\infty$ with $X\le Y$.
\item Cash invariance: $\rho(X+c)=\rho(X)+c $ for $X\in L^\infty$ and $c\in \R$.
\item Convexity: $\rho(\lambda X+ (1-\lambda) Y )\le\lambda \rho(X)+ (1-\lambda) \rho(Y)$ for $X,Y\in L^\infty$ and $\lambda \in [0,1]$. 
\item Positive homogeneity: $\rho(\lambda X)=\lambda \rho(X)$ for $X\in L^\infty$ and $\lambda>0$.
\item Subadditivity: $\rho(X+Y)\le \rho(X)+\rho(Y)$ for $X,Y\in L^\infty$. 
\item Comonotonic additivity: $\rho(X+Y) = \rho(X)+\rho(Y)$ for comonotonic $X,Y\in L^\infty$.  
\item Law invariance: $\rho(X)=\rho(Y)$ when $X\laweq Y$.
\item Quasi-convexity: $\rho(\lambda X+ (1-\lambda) Y )\le\max\{\rho(X), \rho(Y)\}$ for $X,Y\in L^\infty$ and $\lambda \in [0,1]$. 
\end{enumerate}
Following the standard terminology in the literature,  a \emph{monetary risk measure} is a mapping $\rho:L^\infty\to\R$ satisfying (a) and (b);
a \emph{convex risk measure} is a mapping $\rho:L^\infty\to\R$ satisfying (a)--(c);
a \emph{coherent risk measure} is a mapping $\rho:L^\infty\to\R$ satisfying (a)--(d).
Note that convexity (c) and subadditivity (e) are equivalent under positive homogeneity (d).
Any monetary risk measure is $1$-Lipschitz with respect to $L^\infty$-norm, and hence continuous. 
Diversification in the sense of \citet{D89} is represented by quasi-convexity (h); see \citet{CMMM11} for a formal analysis. 
For more background and many properties of risk measures, see \citet{FS16}. 

As explained in the introduction, our main point of interest is the \emph{submodularity} of $\rho$, that is,
 \begin{equation}
 \rho(X\vee Y)+\rho(X\wedge Y)\le \rho(X)+\rho(Y) \mbox{~~~for all $X,Y\in L^\infty$}.   \label{eq:main}
 \end{equation} 
Here, $x\vee y=\max\{x,y\}$ and $x\wedge y=\min\{x,y\}$ for $x,y\in \R$, and they are applied pointwise to random variables. 
Moreover, we say that $\rho$ is \emph{supermodular} if $-\rho$ is submodular. That is, the inequality    in \eqref{eq:main} is reversed.

It is straightforward to verify that submodular functions are closed under some simple operations.
For instance, 
let $\rho,\rho':L^\infty\to\R$ be submodular functions. The mappings $\lambda \rho$ for $\lambda \ge 0$, $\rho+c$ for $c\in \R$, and $\rho+\rho'$ are submodular. Similarly, any convex combination of 
a class  $\{\rho_\theta:\theta \in \Theta\}$ of submodular risk measures is submodular.

We next  review the results of \citet{CC18} that submodular risk measures are closely related to Choquet integrals.
 Let $V$ be the set of mappings $v:\mathcal F\to\R$ with  bounded variation and $v(\varnothing)=0$. Here, by the  variation of $v$ we mean 
 the norm
 $$
 \|v\|=\sup\left\{\sum_{k=1}^n | v(A_k)-v(A_{k-1})|:  n\in \N;~ \varnothing =A_0\subseteq A_1\subseteq\dots\subseteq A_n =\Omega  \right\},
 $$
  which is always finite if $v$ is increasing.  
For $X\in L^\infty$ and  $v\in V$, the \emph{Choquet integral} $ \int X \d v$ is defined as 
$$
\int X \d v= \int_{0}^\infty v (X> x) \d x  +  \int_{-\infty}^0  \left( v(X> x) - v(\Omega) \right)\d x.
$$
Using the terminology of \citet{CC18}, we say that a mapping $\rho:L^\infty\to\R$ is \emph{Choquet} if it  can be written as \begin{equation}\label{eq:Choquet}
\rho(X)=\int X\d v,~~~X\in L^\infty
\end{equation} 
for some $v \in V$. 
When the function $v$ in \eqref{eq:Choquet} is a capacity with $v(\Omega)=1$, $\rho$ is called a \emph{Choquet risk measure} in the terminology of \citet{EMWW21}, which is monetary. 
If a risk measure is law-invariant and Choquet, then it can be written as 
\begin{align}
    \label{eq:disRM}
\rho_\phi(X) = \int X\d (\phi\circ \p)= \int_{0}^\infty  \phi ( \p (X> x)) \d x  +  \int_{-\infty}^0  \left( \phi (\p(X> x)) -\phi(1)\right)\d x
\end{align} for some function $\phi:[0,1]\to \R$ with bounded variation and $\phi(0)=0$. 
When $\phi$ is  further increasing and $\phi(1)=1$, then $\rho_\phi$ is called a \emph{distortion risk measure}, which satisfies the properties (a), (b), (d), (f), and (g) above; any of (c), (e) and (h) holds if and only if $\phi$ is concave; see \citet[Theorem 3]{WWW20}.
As the two most important classes of distortion risk measures in finance and insurance, the VaR of $X$ at level $p\in (0,1)$ is given by $$\VaR_p(X) =\inf\{x\in\R:\p(X\le x)\ge p\},\qquad X\in L^\infty,$$ which is also the left quantile of $X$ at $p$; the ES at level $p\in(0,1)$ is
$$
\ES_p(X)=\frac{1}{1-p}\int_p^1 \VaR_q(X)\,\d q,\qquad X\in L^\infty.
$$ 
Both $\VaR_p$ and $\ES_p$ belong to the class  \eqref{eq:disRM}, with $\phi $ given by $\phi(t)=\id_{\{t>1-p\}}$ and $\phi(t)=\min\{t/(1-p),1\}$, respectively. 
It is well known that any Choquet mapping is comonotonic-additive and norm-continuous.
Moreover, comonotonic additivity and monotonicity characterize monotone Choquet  functionals  (\citealp{S86}).
When $\phi$ is not monotone and it is nonnegative with $\phi(1)=0$, we get 
another special case of \eqref{eq:disRM}, $
    \rho_\phi(X) = \int_{ \R}  \phi ( \p (X> x)) \d x, 
 $   which includes many deviation measures (see Section \ref{sec:4prime}).

  \citet{CC18} showed that   a risk measure $\rho:L^\infty\to\R$  is submodular, positively homogeneous, and monetary if and only if  $\rho$ is a Choquet coherent risk measure. 
 Therefore, for  monetary risk measures that are positively homogeneous, submodularity is fully characterized by   comonotonic additivity and subadditivity.  This result is somewhat surprising, as neither subadditivity nor comonotonic additivity follows from submodularity alone (which will be illustrated by many examples in this paper), and the other   properties also do not imply subadditivity nor comonotonic additivity.
 
 This observation inspired us to wonder what happens when we assume submodularity but not positive homogeneity. 
 \citet{MM08} proved that under a   condition (the hyper-Archimedean property)  on the underlying space, submodularity and cash invariance imply convexity. The required condition  
 is not satisfied by $L^\infty$, but the implication holds when we further impose  $L^\infty$-continuity.
 This observation   is an important point but not explicitly stated in the literature, and hence
we document it here. 
\begin{theorem} 
\label{th:MM08}
   Every submodular monetary risk measure is convex.
\end{theorem}
Theorem \ref{th:MM08} follows essentially from Theorem 15 of \citet{MM08} through a simple approximation technique in Lemma \ref{lem:bounded-extension} below, which will be useful in our other results later.
A sub-$\sigma$-algebra $\mathcal F'$ of $\mathcal F$ is called \emph{simple} if it is generated by finitely many disjoint sets. The corresponding space $L^\infty(\Omega,\mathcal F',\mathbb P)$ is called a \emph{simple subspace}.

\begin{lemma}\label{lem:bounded-extension}
Let $\rho:L^\infty\to\R$ be norm-continuous.
If $\rho$ is submodular (resp.~convex) on every simple subspace, then $\rho$ is submodular (resp.~convex) on all of $L^\infty$.
\end{lemma}
\begin{proof}[Proof]
Fix $X,Y\in L^\infty$. For $n\in\mathbb N$, define
 $
X^{(n)}=2^{-n}\lfloor 2^n X\rfloor$ 
and $
Y^{(n)}=2^{-n}\lfloor 2^n Y\rfloor.
 $
Then
$
\|X^{(n)}-X\|_\infty\le 2^{-n}$ and $
\|Y^{(n)}-Y\|_\infty\le 2^{-n}.
$
As  $X^{(n)}$ and $Y^{(n)}$ are both measurable with respect to the finite $\sigma$-algebra
 $
\mathcal F_n=\sigma(X^{(n)},Y^{(n)}),
 $
  they lie in the same simple subspace $L^\infty(\Omega,\mathcal F_n,\mathbb P)$.
If $\rho$ is submodular on the simple subspace, then 
$
\rho(X^{(n)}\vee Y^{(n)})+\rho(X^{(n)}\wedge Y^{(n)})
\le \rho(X^{(n)})+\rho(Y^{(n)}).
$
If $\rho$ is convex on the simple subspace, then 
$
\rho(\lambda X^{(n)}+ (1-\lambda) Y^{(n)}) 
\le \lambda \rho(X^{(n)})+(1-\lambda)\rho(Y^{(n)})
$ for any $\lambda\in [0,1]$.
Note that 
$
X^{(n)}\vee Y^{(n)}\to X\vee Y $ and $
X^{(n)}\wedge Y^{(n)}\to X\wedge Y$ in  $L^\infty$. 
With $\rho$ being norm-continuous, 
passing to the limit yields
the desired inequality  
$
\rho(X\vee Y)+\rho(X\wedge Y)\le \rho(X)+\rho(Y) $
or $
\rho(\lambda X+(1-\lambda) Y) \le \lambda \rho(X)+(1-\lambda)\rho(Y) $.
\end{proof} 

\begin{proof}[Proof of Theorem \ref{th:MM08}]
Let 
$\rho$ be a monetary risk measure that is submodular. We further assume $\rho(0)=0$ without loss of generality, and this allows us to use the main results of \citet{MM08}.  
 Since \(\rho\) is monetary, it is norm-continuous. By Lemma \ref{lem:bounded-extension}, it suffices to prove convexity on all simple subspaces. 
 For any simple subspace $L^\infty(\Omega,\mathcal F',\mathbb P)$, it is 
hyper-Archimedean as explained in the examples following Theorem 8 in \citet{MM08}.
Therefore, we can apply Theorem 15 of \citet{MM08} to get that $-\rho$ is concave, because it is both supermodular and translation invariant (in the sense that $\rho(X+c)=\rho(X)+c\rho(1)$ for $c\in \R$). This shows that $\rho$ is convex on the simple subspace. Finally, Lemma \ref{lem:bounded-extension} yields convexity on $L^\infty$.
\end{proof}

Note that cash invariance is important for the conclusion in Theorem \ref{th:MM08}; without it we can find many examples of submodular mappings that are not convex (see e.g., Theorem \ref{th:modular-char}).  
Before analyzing specific classes of submodular risk measures, we first discuss the conceptual differences between submodularity and convexity.

\section{Contrast between submodularity and diversification}
\label{sec:3}
Risk measures are often studied with properties such as convexity and subadditivity, which reflect notions of diversification. 
In particular, for law-invariant monetary risk measures, diversification (that is, quasi-convexity) is equivalent to convexity; see \citet{CMMM11}.
Submodularity on its own should be distinguished from diversification, although the two properties coincide for some specific classes of risk measures (see Proposition \ref{prop:CE} and Theorem \ref{th:distortion} below).  For monetary risk measures, Theorem \ref{th:MM08} shows that submodularity is stronger than convexity, but in general there is no implication in either direction.
Indeed,  a functional of the form $X\mapsto \E[\ell(X)]$ is submodular for any function $\ell$ (Theorem \ref{th:modular-char} below), but the preferences represented by this functional do not necessarily exhibit diversification (and it is not necessarily quasi-convex).  

\begin{figure}[t]
\centering
\resizebox{0.97\textwidth}{!}{\includegraphics[trim=0 38 0 0, clip]{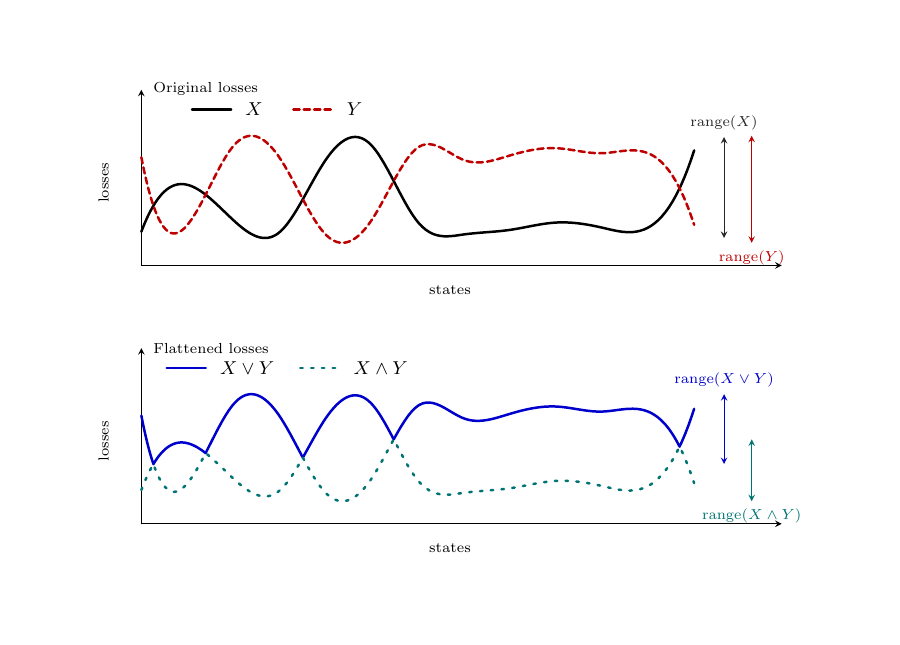}} 
\caption{A stylized example for the flattening intuition behind submodularity. Here $X\vee Y$ and $X\wedge Y$ (bottom panel) fluctuate less than $X$ and $Y$ (top panel) and have smaller ranges.}
\label{fig:flattening}
\end{figure}

To interpret submodularity for risk measures, we should compare $X$ and $Y$ with $X\wedge Y$ and $X\vee Y$.
Figure \ref{fig:flattening} gives a stylized example in which  we can see that   $X\wedge Y$ and $X\vee Y$  are ``flatter" than the original positions $X$ and $Y$ in the sense that their ranges are smaller. {Indeed, writing $a,b$ for the ranges of $X$ and $Y$ with $a\ge b$
and $c,d$ for the ranges of $X\wedge Y$ and  $X\vee Y$  with $c\ge d$,  we always have $a\ge c$ and $b\ge d$.} Hence, we may call $X\wedge Y$ and $X\vee Y$ the flattened positions.  
Note that submodularity means that the total risk 
of the flattened positions 
is smaller than the original positions.  
Therefore, submodularity reflects the following intuitive principle:
\begin{center}
\emph{It is rewarding to flatten the risks in a specific way.}
\end{center}

The simple (but extreme) case where   $X\wedge Y$ and 
$X\vee Y$ are both constants is   illustrative. 
For instance, we can think of $X=\id_A$ and $Y=\id_{A^c}$ for some event $A$. 
Note that $X$ and $Y$ are random, but $z:=X\wedge Y$ and $w:=X\vee Y$ are deterministic.
Submodularity, reading as
\begin{align}
\label{eq:econ-int}\rho(z)+\rho(w)\le \rho(X)+\rho(Y),
\end{align}
means that the total risk of two constant losses 
is smaller than the total risk of two random losses, while $z+w= X+Y$.
 Note that \eqref{eq:econ-int} is quite natural for a risk measure. For instance, it is satisfied by any monetary risk measure $\rho$  with $\rho(0)=0$ and $\rho \ge \E$ (this includes all law-invariant coherent risk measures), because
$$
\rho(z)+\rho(w)=z+w=\E[X]+\E[Y]\le \rho(X) + \rho(Y). 
$$
Here, the interpretation for rewarding the flattened positions is quite clear. 


To further illustrate the sharp contrast between submodularity and diversification, we state the following simple result that follows essentially from  the definition. 
For a measurable function $f:\R\to \R$, 
we define $\rho \circ f$
as the mapping $X\mapsto \rho(f(X))$, which is a minor abuse of notation. 

\begin{theorem}
\label{th:contrast}
Let $f:\R\to \R$ be an increasing function and $\rho:L^\infty\to \R$.
    If $\rho$ is quasi-convex, then so is $f\circ \rho$, but not necessarily $\rho\circ f$.
    If $\rho$ is submodular, then so is $\rho\circ f$,  but not necessarily $f\circ \rho$.
\end{theorem}
\begin{proof}
    For the first statement, the quasi-convexity of $f\circ \rho$ is straightforward from the definition. To see that $\rho\circ f$ may fail to be quasi-convex, 
let $\rho=\E$ (which is convex)
   and $f$ be a strictly concave function on $\R$. For two a.s.~distinct random variables $X$ and $Y$ with $\E[f(X)]=\E[f(Y)]$, we have $\E[f(X/2+Y/2)]>\E[f(X)]/2 + \E[f(Y)]/2 = \E[f(X)]\vee \E[f(Y)]$, violating quasi-convexity.

  For the second statement, we first show the submodularity of $\rho\circ f$. By the submodularity of $\rho$ and the monotonicity of $f$, we have
   $$
      \rho(f(X\vee Y))+
      \rho(f(X\wedge Y))
      =
   \rho(f(X)\vee f(Y))+
      \rho(f(X)\wedge f(Y))
   \le   \rho(f(X) ) + \rho( f(Y)),
   $$
   and therefore $\rho\circ f$ is submodular. 
   To see that $f\circ \rho$ is not necessarily submodular,   
   let $\rho=\E$ (which is submodular)
   and $f$ be a strictly convex function on $\R$. For two a.s.~distinct random variables $X$ and $Y$ with $\E[X]=\E[Y]$, 
   we have 
   $$
   f(\E[X\vee Y]) + f(\E[X\wedge Y])
   >  2 f\left(\frac{\E[X\vee Y] + \E[X\wedge Y]}{2}\right)
   =   f(\E[X]) + f(\E[Y]),
   $$
   violating submodularity. 
\end{proof}

Theorem \ref{th:contrast} 
yields a simple invariance property of submodularity. 
It also
shows that 
diversification (represented by quasi-convexity) and submodularity are  
fundamentally different, from an operational perspective. 

\begin{remark}
 Theorem \ref{th:contrast} implies that for any given  strictly increasing function $f:\R\to\R$, $\rho$ is quasi-convex if and only if $f\circ \rho$ is quasi-convex; 
    $\rho$ is submodular if and only if $  \rho\circ f$ is submodular. 
\end{remark}

\begin{remark}
A classic interpretation for 
  submodularity in economics
  is complementarity in the terminology of \citet{T78}, which may be less transparent in the context of risk measures. 
 Complementarity concerns how the marginal increase of risk from a higher loss in one scenario changes when another scenario is already worse.
To illustrate this, consider a finite sample  space $\Omega=\{\omega_1,\dots,\omega_n\}$, and represent $X$
by its scenario values $(X(\omega_1),\dots,X(\omega_n))$.
For $\epsilon,\delta\ge 0$, submodularity implies
\begin{align}
\rho\!\begin{pmatrix}
X(\omega_1)+\epsilon\\
X(\omega_2)+\delta\\
X(\omega_3)\\
\vdots\\
X(\omega_n)
\end{pmatrix}
-\rho\!\begin{pmatrix}
X(\omega_1)\\
X(\omega_2)+\delta\\
X(\omega_3)\\
\vdots\\
X(\omega_n)
\end{pmatrix}
\le \rho\!\begin{pmatrix}
X(\omega_1)+\epsilon\\
X(\omega_2)\\
X(\omega_3)\\
\vdots\\
X(\omega_n)
\end{pmatrix}
-\rho\!\begin{pmatrix}
X(\omega_1)\\
X(\omega_2)\\
X(\omega_3)\\
\vdots\\
X(\omega_n)
\end{pmatrix}.\label{eq:explain}
\end{align}
Moreover, we can see that  repeatedly applying \eqref{eq:explain}  with permutations yields submodularity. 
Interpreting \eqref{eq:explain}, 
we see that the marginal risk increase of loss $\epsilon$ added in scenario $\omega_1$ becomes larger when more loss occurs in scenario $\omega_2$.

\end{remark}

\section{Expected losses and certainty equivalents}
\label{sec:4}

\subsection{Expected losses}
An expected loss is the mapping  
\begin{align}
\label{eq:EL}
E_\ell(X)=\E[\ell(X)], ~~~~X\in L^\infty,
\end{align}
 where $\ell:\R\to\R$ is a measurable function, which we call a loss function. 
This is analogous to the expected utility for a utility function $u:\R\to\R$ in decision theory, with the difference that our random variables in $L^\infty$ represent losses.
Note that $\ell(X)$ is not necessarily in the same unit as $X$.
It is straightforward to check from the  definition that $E_\ell$ is modular (both submodular and supermodular). 
The following result gives a simple characterization of all law-invariant functionals that are modular. 
\begin{theorem}\label{th:modular-char}
Let $(\Omega,\mathcal F,\p)$ be an atomless probability space and $\rho:L^\infty\to\R$ be law-invariant and
$\|\cdot\|_\infty$-continuous.
Then the following are equivalent:
\begin{enumerate}[(i)]
\item $\rho$ is both submodular and supermodular.
\item There exists a continuous function $\ell:\R\to\R$  such that
$\rho=E_\ell$.
\end{enumerate}
\end{theorem}

\begin{proof}
(ii)$\Rightarrow$(i): 
The pointwise identity $\ell(a)+\ell(b)=\ell(a\vee b)+\ell(a \wedge b)$ gives
$E_\ell(X)+E_\ell(Y)=E_\ell(X\vee Y)+E_\ell(X\wedge Y)$.

(i)$\Rightarrow$(ii):
Since $\rho$ is both submodular and supermodular, it is called a \emph{valuation}:
\begin{equation}\label{eq:valuation}
\rho(X)+\rho(Y)=\rho(X\vee Y)+\rho(X\wedge Y),\qquad \forall\,X,Y\in L^\infty.
\end{equation}

Define $\ell(x)=\rho(x\id_\Omega)$. Without loss of generality, assume  $\rho(0)=\ell(0)=0$.  Note that $\ell$ is continuous since $\rho$ is $\|\cdot\|_\infty$-continuous and $\|x\id_\Omega - y\id_\Omega\|_\infty = |x-y|$.

Fix $a>b$. For $t\in[0,1]$, choose $A\in\mathcal F$ with $\p(A)=t$ and set
$$
\psi(t)=\rho(a\id_A+b\id_{A^c}),
$$
which is well-defined by law-invariance. For $s,t\ge 0$ with $s+t\le 1$, choose disjoint $A,B$ with $\p(A)=t$, $\p(B)=s$. Setting $X=a\id_A+b\id_{A^c}$ and $Y=a\id_B+b\id_{B^c}$, we have $X\wedge Y=b$ and $X\vee Y=a\id_{A\cup B}+b\id_{(A\cup B)^c}$, so \eqref{eq:valuation} gives
\begin{equation}\label{eq:cauchy}
\varphi(t+s)=\varphi(t)+\varphi(s), \qquad s,t\ge 0,\; s+t\le 1,
\end{equation}
where $\varphi(t)=\psi(t)-\psi(0)$. The family $\{a\id_A+b\id_{A^c}:A\in\mathcal F\}$ is order-bounded in $L^\infty$; indeed, with $m=|a|\vee |b|$,
$$
-m\id_\Omega\le a\id_A+b\id_{A^c}\le m\id_\Omega,\qquad A\in\mathcal F,
$$
and therefore
$$
\|a\id_A+b\id_{A^c}\|_\infty\le |a|\vee |b|,\qquad A\in\mathcal F.
$$
Since $\rho$ is a $\|\cdot\|_\infty$-continuous valuation on the Banach lattice $L^\infty$, it is bounded on order-bounded sets \cite[Proposition~3.2]{TV20}. Hence $\psi$ is bounded on $[0,1]$, and so is $\varphi$. Let $M=\sup_{u\in[0,1]}|\varphi(u)|<\infty$. For any $u\in[0,1/n]$, \eqref{eq:cauchy} gives $\varphi(nu)=n\varphi(u)$ and thus $|\varphi(u)|\le M/n$, so $\varphi$ is continuous at $0$. By \eqref{eq:cauchy}, continuity at $0$ implies continuity on $[0,1]$, and $\varphi(q)=q\varphi(1)$ for rational $q\in[0,1]$. Therefore, by density of rationals, $\varphi(t)=t\varphi(1)$ for all $t\in[0,1]$. Equivalently, $\varphi$ is additive and bounded on an interval, and hence linear on $[0,1]$. Therefore
\begin{equation}\label{eq:twopoint}
\rho(a\id_A+b\id_{A^c})=(1-t)\ell(b)+t\ell(a)=E_\ell(a\id_A+b\id_{A^c}).
\end{equation}

Now let $X=\sum_{i=1}^n x_i\id_{A_i}$ be simple with $x_1>\cdots>x_n$ and $t_i=\p(A_i)$. We show by induction on $n$ that
\begin{equation}\label{eq:simple}
\rho(X)=E_\ell(X).
\end{equation}
The case $n=1$ is the definition of $\ell$, and $n=2$ is \eqref{eq:twopoint}. For the inductive step, define $Y=x_2\id_{A_1}+x_1\id_{A_1^c}$ (a two-valued random variable). Then $X\vee Y=x_1$ and $X\wedge Y=x_2\id_{A_1\cup A_2}+\sum_{i=3}^n x_i\id_{A_i}$, which takes $n-1$ distinct values. Applying \eqref{eq:valuation},
$$
\rho(X)=\ell(x_1)+\rho(X\wedge Y)-\rho(Y).
$$
By the induction hypothesis and \eqref{eq:twopoint},
$$
\rho(X\wedge Y)=(t_1+t_2)\ell(x_2)+\sum_{i=3}^n t_i\ell(x_i),\qquad
\rho(Y)=(1-t_1)\ell(x_1)+t_1\ell(x_2),
$$
and substituting yields \eqref{eq:simple}.

Finally, for general $X\in L^\infty$, let $X_k=2^{-k}\lfloor 2^k X\rfloor$. Then $\|X_k-X\|_\infty\le 2^{-k}$, so $\rho(X_k)\to\rho(X)$ by $\|\cdot\|_\infty$-continuity. By \eqref{eq:simple}, $\rho(X_k)=E_\ell(X_k)$. Since $\ell(X_k)\to\ell(X)$ almost surely and is uniformly bounded, dominated convergence gives $E_\ell(X_k)\to E_\ell(X)$, hence $\rho(X)=E_\ell(X)$.
\end{proof}

\subsection{Certainty equivalents with respect to expected losses}
\label{sec:CE}
The certainty equivalents are a popular object in decision theory that play a similar role to risk measures. They are not monetary risk measures as they do not satisfy cash invariance in general, but they satisfy the property $\rho(c)=c$ for all $c\in \R$, so that the input and output are both the monetary scale.
 
Let $\ell:\R\to\R$ be a strictly increasing function, and define its generalized inverse by
$$
\ell^{-1}(y)=\inf\{x\in\R:\ell(x)\ge y\},\qquad y\in\R.
$$
The certainty equivalent (CE) with respect to the   loss function $\ell$ is given by the mapping
\begin{align}
\label{eq:CE-loss}
\mathrm{CE}_\ell (X)=
\ell^{-1}(\E[\ell(X)]) ,~~~~X\in L^\infty.
\end{align}

\begin{proposition}\label{prop:CE}
The certainty equivalent functional $\mathrm{CE}_\ell$ in \eqref{eq:CE-loss} is submodular if and only if $\ell$ is convex. 
\end{proposition}
\begin{proof}
For the ``if'' direction, assume that $\ell$ is convex. Since $\ell$ is strictly increasing, its generalized inverse $\ell^{-1}$ is concave. The submodularity of $\mathrm{CE}_\ell $ is equivalent to
\begin{equation}
\label{eq:CE1}
\ell^{-1}(\underbrace{\E[\ell(X \vee Y)]}_{=a})
+\ell^{-1}(\underbrace{\E[\ell(X \wedge Y)]}_{=b})
\le 
\ell^{-1}(\underbrace{\E[\ell(X)]}_{=c})
+\ell^{-1}(\underbrace{\E[\ell(Y)]}_{=d})
\end{equation}
for all $X,Y\in L^\infty$. Because $a+b=c+d$, and $\ell$ is increasing with
$X\vee Y\ge X$, $X\vee Y\ge Y$, $X\wedge Y\le X$, and $X\wedge Y\le Y$, we have
$a\ge c$, $a\ge d$, $b\le c$, and $b\le d$. Hence $c,d\in [b,a]$, and
in dimension two this is equivalent to $(a,b)$ majorizing $(c,d)$, since
$a+b=c+d$ and $\max\{a,b\}=a\ge \max\{c,d\}$,
the concavity of $\ell^{-1}$ guarantees \eqref{eq:CE1}. For the ``only if'' direction, suppose that \eqref{eq:CE1} holds. 
For any $a\ge b$ in the 
range of $\ell$,
we can write $a=\ell(x)$ and $b=\ell(y)$ for some $x,y\in \R$ with $x\ge y$. 
Let $X=x\id_{A}+y\id_{A^c}$
and $Y=y\id_{A}+x\id_{A^c}$
where $A$ is an event with $\p(A)=1/2$.
Then we have $X\vee Y=x$, $X\wedge Y=y$, and $\E[\ell(X)]=\E[\ell(Y)]=(a+b)/2$.
Hence,
\eqref{eq:CE1} becomes
$$
\ell^{-1}( a  )
+\ell^{-1}( b )
\le 
2\ell^{-1}\left( \frac{a+b}{2}\right),
$$
and thus $\ell^{-1}$ is midpoint concave on the range of $\ell$.
Since $\ell^{-1}$ is increasing on $\R$, it is bounded on every closed interval and hence locally bounded. A midpoint-concave and locally bounded function on an interval is concave. Therefore, $\ell^{-1}$ is concave, and $\ell$ is convex.
\end{proof}
For this family, submodularity is characterized by convexity of the underlying loss function $\ell$. Specifically, $\mathrm{CE}_\ell$ is submodular if and only if $\ell^{-1}$ is concave, which is equivalent to $\ell$ being convex.

 
 \section{Distortion risk measures and deviation measures}
 \label{sec:4prime}
\subsection{Distortion risk measures}


We now study submodularity for distortion risk measures and coherent risk measures. 
For 
a distortion risk measure  $\rho_\phi$ in \eqref{eq:disRM}, coherence holds  if and only if $\phi$ is concave, and it is also equivalent to several other properties, such as convex-order consistency, subadditivity, quasi-convexity, and concavity on mixtures; see, for example, Theorem~3 of \citet{WWW20}. 
 
Using the main result of \citet{CC18} and known characterizations of distortion risk measures due to \citet{Y87}, we arrive at the following characterization. The proof of the result is simple, but it offers several new characterizations of coherent distortion risk measures.
\begin{theorem}\label{th:distortion}
 For a law-invariant risk measure $\rho:L^\infty\to\R$, 
the following statements are equivalent:
\begin{enumerate}[(i)]
\item  $\rho$ is  coherent and submodular; 
\item  $\rho$ is  positively homogeneous, monetary, and submodular; 
\item $\rho$ is  submodular, comonotonic-additive, and monetary; 
\item $\rho$ is  convex (or quasi-convex), comonotonic-additive, and monetary; 
\item $\rho$ is a distortion risk measure with a concave distortion function; 
\item $\rho$ is a submodular distortion risk measure. 
\end{enumerate} 
\end{theorem}

\begin{proof}
In each implication below, we use only the properties explicitly assumed in that implication (plus law invariance from the theorem statement).
 
(i)$\Rightarrow$(ii):
Immediate by definition.
 
(ii)$\Rightarrow$(iii):
A monetary and positively homogeneous functional on $L^\infty$ is Lipschitz, hence  $L^\infty$-continuous. Together with submodularity, Theorem~2.2 of \citet{CC18} gives a Choquet integral representation, which satisfies comonotonic additivity. 

(iii)$\Rightarrow$(iv): Corollary~2.3 of \citet{CC18} shows that submodularity and subadditivity coincide for Choquet integrals.
Subadditivity is equivalent to convexity for positively homogeneous functionals. Convexity is equivalent to quasi-convexity for monetary risk measures \citep[Proposition 2.1]{CMMM11}.

(iv)$\Rightarrow$(v): This follows directly from the standard characterization of distortion risk measures; see \citet{S86} and \citet{Y87}.

(v)$\Rightarrow$(vi): It is well known that convexity of a distortion risk measure is equivalent to concavity of the distortion, and submodularity follows from Corollary~2.3 of \citet{CC18}.

(vi)$\Rightarrow$(i):  This follows from Corollary~2.3 of \citet{CC18}.
\end{proof}

Theorem~\ref{th:distortion} concludes the case of law invariance and coherence: submodularity is equivalent to   being a distortion risk measure in this setting.  
In Section \ref{sec:5} we will focus on risk measures without positive homogeneity, where the situation becomes more delicate.

\subsection{Deviation measures}

We next analyze various deviation measures introduced by \citet{RUZ06}.
Four commonly used examples,  the variance, the mean-median deviation, the range, and the Gini deviation  are defined by, respectively, for $X\in L^\infty$,
\begin{equation}
\label{eq:devms}
\begin{aligned}
\var(X)&=\E\bigl[(X-\E[X])^2\bigr],\\
\MMD(X)&=\min_{m\in\mathbb{R}}\E\bigl[|X-m|\bigr]
\;=\E\bigl[|X-\VaR_{1/2}(X)|\bigr],\\
\mathrm{range}(X)&=\esssup X-\essinf X,\\
\mathrm{Gini}(X)&=\frac{1}{2}\,\E\bigl[|X-X'|\bigr].
\end{aligned}
\end{equation}

where $\esssup X$ and $\essinf X$ 
are the essential supremum and essential infimum of $X$, respectively,  and $X'$ is an independent copy of $X$. 
These examples belong to two general classes. 
The first class contains the mappings  $    \nu_\psi$ defined by 
\begin{align}
    \label{eq:deviation1}
    \nu_\psi(X)= \E [\psi(X-X')],~~~X\in L^\infty,
\end{align} 
for convex functions $\psi:\R\to \R$, where $X'$ is an independent copy of $X$. The class  \eqref{eq:deviation1} includes  the variance and the Gini deviation as special cases (with
 $\psi(u)=u^2/2$ and $\psi(u)=|u|/2$, respectively). 
The second class contains the Choquet mappings $\rho_\phi$  given by 
\begin{equation}
\label{eq:deviation2}
\rho_{\phi}(X)=\int_\R \phi(\p(X>x))\d x,~~~X\in L^\infty,
\end{equation}
where $\phi$ is a concave function that satisfies $\phi(0)=\phi(1)=0$.
The class \eqref{eq:deviation2} includes the mean-median deviation ($\phi(t)=t\wedge (1-t)$), the range ($\phi(t)=\id_{\{0<t<1\}}$), and the Gini deviation ($\phi(t)=t(1-t)$; in both classes); see Table 1 of \citet{WWW20a}.

\begin{proposition}
\label{prop:dev}
The mappings $\nu_\psi$ in \eqref{eq:deviation1} for convex $\psi$
and $\rho_\phi$  in \eqref{eq:deviation2} for concave $\phi$
 are submodular.  
 As a consequence, the mappings in \eqref{eq:devms} are submodular. 
\end{proposition}

\begin{proof}
We first show the submodularity of  the class $\nu_\psi$ in \eqref{eq:deviation1}. 
Define
$$
g_\psi(s,t)=\psi(s-t),\qquad s,t\in\mathbb{R}.
$$
It is straightforward to verify that $g_\psi$ is submodular on $\mathbb{R}^2$. 
Take $X,Y\in L^\infty$, and let $(X',Y')$ be an independent copy of $(X,Y)$. Applying the above inequality pointwise to $(X,X')$ and $(Y,Y')$ gives
$$
\psi\bigl((X\vee Y)-(X'\vee Y')\bigr)+\psi\bigl((X\wedge Y)-(X'\wedge Y')\bigr)
\le \psi(X-X')+\psi(Y-Y').
$$
Hence the functional $\nu_\psi$
is submodular. 

For the class $\rho_\phi$ in \eqref{eq:deviation2},
first note that it is a convex Choquet mapping and it is norm-continuous; see Theorems~1 and~3 of \citet{WWW20}. By Corollary~2.3 of \citet{CC18}, we get that $\rho_\phi$ is submodular.
\end{proof}

Although the two very general classes of deviation measures in \eqref{eq:deviation1}--\eqref{eq:deviation2} are submodular, we note that not all classical deviation measures are submodular. In particular, the standard deviation, denoted by $\mathrm{SD}=\sqrt{\var}$, is not submodular, as shown by the next example.
\begin{example}\label{ex:SD}
    Take \(Y= 0\) and $X$ with $\p(X=-2)=\p(X=2)=1/4$  and $\p(X=0)=1/2$. 
Then \(\E[X]=0\) and \(\E[X^2]=2\), so \(\mathrm{SD}(X)=\sqrt 2\).
It is straightforward to compute $\mathrm{SD}(X\vee 0)=\mathrm{SD}(X\wedge 0)=\sqrt 3/2$. 
Hence, $
\mathrm{SD}(X\vee 0)+\mathrm{SD}(X\wedge 0)=\sqrt3 >\sqrt2 =\mathrm{SD}(X)+\mathrm{SD}(0),
 $
and thus the standard deviation is not submodular.
\end{example}
As shown by Theorem~2.1 of \citet{CC18}, a mapping $\rho$ on a finite space satisfying submodularity, positive homogeneity, and translation invariance (that is, $\rho(X+c)=\rho(X)+c$ for all constants $c$) must be Choquet. Since $\mathrm{SD}$ does not admit a Choquet representation on finite spaces, it cannot be submodular by using that result; Example \ref{ex:SD} provides a direct verification.
 
 \section{Four classes of convex risk measures}
 \label{sec:5}
 
 As we have seen from   Theorem \ref{th:distortion}, coherent risk measures that are submodular 
 are well understood and fully characterized.
The picture for submodular monetary risk measures is much less complete, although we know from
 Theorem \ref{th:MM08}  that they are convex. 
 In this section, we consider four classes of convex risk measures that admit explicit formulas. 
These four classes are quite general and include most commonly used convex risk measures.

\subsection{Shortfall risk measures}
Let $\ell:\R\to(-\infty,\infty]$ be a strictly increasing convex function, which we call a loss function.
A \emph{shortfall risk measure} as in \citet{FS02} with loss function $\ell$ is defined by
$$
\rho_\ell(X)=\inf\{m\in\R:\E[\ell(X-m)]\le \ell(0)\},\qquad X\in L^\infty.
$$
Because $\ell$ is continuous and strictly increasing,
$\rho_\ell(X)$ satisfies
\begin{equation}\label{eq:rho}
\E[\ell(X-\rho_\ell(X))]=\ell(0).
\end{equation}
Now replace $\ell$ by
$$
\tilde\ell(x)=\ell(x)-\ell(0),\qquad x\in\R.
$$
Then $\tilde\ell$ is still strictly increasing and convex, and
$$
\E[\ell(X-m)]\le \ell(0)\iff \E[\tilde\ell(X-m)]\le 0.
$$
Hence, this normalization does not change the shortfall risk measure and we may assume $\ell(0)=0$ without loss of generality. Under this normalization, \eqref{eq:rho} reads as
$
\E[\ell(X-\rho_\ell(X))]=0.
$

In the next result, we consider the simpler case where $\ell\in C^2(\R)$, and let
$$
R(x)=\frac{\ell''(x)}{\ell'(x)}~ \mbox{for} ~x\in\R,~~~ \mbox{and}~~~
L=\inf_{y\in\R}R(y).
$$ 
Note that $R$ is nonnegative since $\ell$ is convex. 
The general case where $\ell$ is convex but not necessarily in $C^2(\R)$ will be treated in Section \ref{app:SR-proof} of Appendix \ref{app:proofs}, which requires more involved analysis. 
 
\begin{theorem}\label{thm:linear-iff}
Assume $\ell\in C^2(\R)$ is strictly increasing and convex with $\ell(0)=0$. Then the shortfall risk measure $\rho_\ell$ is submodular if and only if there exists $\lambda\in\R$ such that
\begin{equation}\label{eq:LD}
R(x)-2 L   \le \lambda\,\frac{\ell(x)}{\ell'(x)},\qquad \mbox{ for all }x\in\R.
\end{equation}
\end{theorem}

\begin{proof}
Set
$$
S(x)=\ell'(x),\qquad
h(x)=S(x)\big(R(x)-2L\big),\qquad x\in\R.
$$
Since $\ell$ is strictly increasing and convex, $S(x)>0$ for all $x$, so \eqref{eq:LD} is equivalent to
\begin{equation}\label{eq:LDg}
h(x)\le \lambda\,\ell(x),\qquad x\in\R.
\end{equation}

\textbf{Sufficiency.}
Assume \eqref{eq:LDg}. It is enough to show that $\rho_\ell$ is submodular on every simple subspace and then apply Lemma~\ref{lem:bounded-extension}.

Let $\mathcal F'=\sigma(A_1,\dots,A_n)$ be a simple sub-$\sigma$-algebra, where $A_1,\dots,A_n$ are its atoms and
$$
\pi_k=\mathbb P(A_k)>0,\qquad \sum_{k=1}^n \pi_k=1.
$$
For
$$
X=\sum_{k=1}^n x_k\mathbf \id_{A_k},
$$
set $m(\bx)=\rho_\ell(X)$. Then
\begin{equation}\label{eq:implicit-weighted}
\sum_{k=1}^n \pi_k\,\ell(x_k-m(\bx))=0.
\end{equation}
Write $y_k=x_k-m(\bx)$ and
$$
T(\bx)=\sum_{k=1}^n \pi_k S(y_k),\qquad
w_k(\bx)=\frac{\pi_k S(y_k)}{T(\bx)}.
$$
Since $T(\bx)>0$, the implicit function theorem gives $m\in C^2(\R^n)$, and
$$
\partial_j m(\bx)=w_j(\bx).
$$
For $i\ne j$, a direct calculation gives
\begin{equation}\label{eq:cross-weighted}
\partial_{ij}^2 m(\bx)
=-w_iw_j\Big(R(y_i)+R(y_j)-\sum_{k=1}^n w_k R(y_k)\Big).
\end{equation}

Evaluating \eqref{eq:LDg} at $y_k$ and summing with weights $\pi_k$ yields
$$
\sum_{k=1}^n \pi_k S(y_k)\big(R(y_k)-2L\big)
\le \lambda \sum_{k=1}^n \pi_k \ell(y_k)=0,
$$
where the last equality follows from \eqref{eq:implicit-weighted}. Dividing by $T(\bx)>0$, we obtain
$$
\sum_{k=1}^n w_k R(y_k)\le 2L.
$$
Since $R(y_i)+R(y_j)\ge 2L$, \eqref{eq:cross-weighted} implies $\partial_{ij}^2m(\bx)\le 0$ for all $i\ne j$. Therefore $m$ is submodular on $\R^n$, and hence $\rho_\ell$ is submodular on $L^\infty(\Omega,\mathcal F',\mathbb P)$.

Since $\mathcal F'$ was arbitrary, $\rho_\ell$ is submodular on every simple subspace. Lemma~\ref{lem:bounded-extension} now yields submodularity on all of $L^\infty$.

\textbf{Necessity.}
Assume $\rho_\ell$ is submodular on $L^\infty$. In particular, it is submodular on every uniform $n$-atom space. Fix $n\ge 3$ and work on an $n$-atom space with equal weights.
Identify $X$ with $\bx=(x_1,\dots,x_n)\in\R^n$ and define $m(\bx)=\rho_\ell(X)$. Then
\begin{equation}\label{eq:implicit}
\sum_{k=1}^n \ell(x_k-m(\bx))=0.
\end{equation}
Write $y_k=x_k-m(\bx)$ and
$$
T(\bx)=\sum_{k=1}^n S(y_k),\qquad
w_k(\bx)=\frac{S(y_k)}{T(\bx)}.
$$
Since $\sum_k S(x_k-m)>0$, the implicit function theorem gives $m\in C^2(\R^n)$ and
$$
\partial_jm(\bx)=w_j(\bx).
$$
A second differentiation, using $S'=SR$ and $\partial_j y_k = \mathbf{1}_{\{k=j\}} - w_j$, gives for $i\ne j$
\begin{equation}\label{eq:cross}
\partial_{ij}^2 m(\bx)
=-w_iw_j\Big(R(y_i)+R(y_j)-\sum_{k=1}^n w_k R(y_k)\Big).
\end{equation}
Since $\partial_{ij}^2 m\le 0$ and $w_i,w_j>0$, we obtain
\begin{equation}\label{eq:weighted-ineq}
\sum_{k=1}^n w_k(\bx)\,R(y_k)\le R(y_i)+R(y_j).
\end{equation}

We first establish a balanced-sum inequality. Fix $\varepsilon>0$ and pick $v$ with $R(v)\le L+\varepsilon$, which exists by definition of $L=\inf R$.
Fix any $x_1,\dots,x_r$ with $\sum_{k=1}^r\ell(x_k)=0$.
Choose $p\in\R$ with $\ell(p)\neq0$ and $\ell(p)\ell(v)<0$, and set
$$
M=\left\lfloor\frac{-2\ell(v)}{\ell(p)}\right\rfloor,\qquad
d=-2\ell(v)-M\ell(p),
$$
so that $d\,\ell(p)\ge0$ and $|d|<|\ell(p)|$.
By continuity and strict monotonicity of $\ell$, choose $c\in\R$ with $\ell(c)=d$.
For each $N\in\mathbb N$, form $\bz^{(N)}\in\R^n$ ($n=M+Nr+3$) consisting of two copies of $v$, $M$ copies of $p$, one copy of $c$, and $N$ copies of $(x_1,\dots,x_r)$. Then $\sum_k\ell(z_k^{(N)})=0$, so $m(\bz^{(N)})=0$.

Apply \eqref{eq:weighted-ineq} at $\bz^{(N)}$ with $(i,j)$ indexing the two $v$-coordinates. Multiplying both sides by $K=\sum_k S(z_k^{(N)})$ gives
$$
\sum_{k=1}^n S(z_k^{(N)})\,R(z_k^{(N)})\le 2R(v)\,K.
$$
Rearranging and isolating the $N$ repeated blocks,
$$
N\sum_{k=1}^r S(x_k)\big(R(x_k)-2R(v)\big)\le 2S(v)R(v) - A,
$$
where $A=M S(p)\big(R(p)-2R(v)\big)+S(c)\big(R(c)-2R(v)\big)$ is independent of $N$. Dividing by $N$ and letting $N\to\infty$ yields
$$
\sum_{k=1}^r S(x_k)\big(R(x_k)-2R(v)\big)\le 0.
$$
Since $R(v)\le L+\varepsilon$,
$$
\sum_{k=1}^r h(x_k)
=\sum_{k=1}^r S(x_k)\big(R(x_k)-2L\big)
\le 2\varepsilon\sum_{k=1}^r S(x_k).
$$
Sending $\varepsilon\downarrow0$:
\begin{equation}\label{eq:balanced}
\sum_{k=1}^r h(x_k)\le 0
\quad\text{whenever}\quad
\sum_{k=1}^r \ell(x_k)=0.
\end{equation}

It remains to deduce \eqref{eq:LDg} from \eqref{eq:balanced}. Fix $a,b$ with $\ell(a)<0<\ell(b)$ and set
$$
\theta=\frac{\ell(b)}{\ell(b)-\ell(a)}\in(0,1),\qquad
\theta\ell(a)+(1-\theta)\ell(b)=0.
$$
Let $r_N=\lceil\theta N\rceil$ and $s_N=N-r_N$. The residual
$$
\delta_N=-(r_N\ell(a)+s_N\ell(b))
$$
satisfies $|\delta_N|\le|\ell(a)-\ell(b)|$ since $0\le r_N-\theta N<1$, so $\delta_N/N\to 0$. Choose $c_N$ with $\ell(c_N)=\delta_N$; since $\delta_N$ is bounded, $c_N$ stays in a compact set, and by continuity of $h$, $h(c_N)/N\to0$. Now
$$
r_N\ell(a)+s_N\ell(b)+\ell(c_N)=0,
$$
so \eqref{eq:balanced} applied to $r_N$ copies of $a$, $s_N$ copies of $b$, and $c_N$ gives
$$
r_N h(a)+s_N h(b)+h(c_N)\le 0.
$$
Dividing by $N$ and letting $N\to\infty$:
$$
\theta\, h(a)+(1-\theta)\,h(b)\le 0
\;\Longleftrightarrow\;
\frac{h(b)}{\ell(b)}\le \frac{h(a)}{\ell(a)}.
$$
Hence
$$
\alpha^+=\sup_{\ell(x)>0}\frac{h(x)}{\ell(x)}
\le
\inf_{\ell(x)<0}\frac{h(x)}{\ell(x)}=\alpha^-.
$$
Both sides are finite: fixing any $a$ with $\ell(a)<0$ shows $\alpha^+\le h(a)/\ell(a)<\infty$, and similarly for $\alpha^-$. Any $\lambda\in[\alpha^+,\alpha^-]$ satisfies $h(x)\le \lambda\ell(x)$ for all $x$ with $\ell(x)\ne0$. Taking all $x_k=0$ in \eqref{eq:balanced} gives $h(0)\le 0=\lambda\ell(0)$, completing the proof of \eqref{eq:LDg}.
\end{proof}

Shortfall risk measures form a broad  family of convex risk measures, but submodularity does not hold in general in this family. 
Theorem~\ref{thm:linear-iff} shows that the curvature $R$ is important for submodularity. 
A special case of \eqref{eq:LD} is the case of $\lambda=0$, which yields the   condition  
\begin{equation}
    \label{eq:suff}
\sup_{x\in \R} R(x)  \le 2 \inf_{x\in \R} R(x).  
\end{equation}
This is sufficient but not necessary for submodularity, and it is very easy to check.

Since loss functions $\ell$ can be converted into utility functions $u$ via 
$\ell(x)=-u(-x)$, we can verify  $$R(x)=\frac{-u''(-x)}{u'(-x)}=A(-x),$$ where $A$ is the Arrow--Pratt coefficient of absolute risk aversion (AP coefficient).
Therefore, the condition in Theorem~\ref{thm:linear-iff} can be seen as a structural restriction on the AP coefficient compatible with submodularity. 
To interpret the sufficient condition \eqref{eq:suff}, 
it means that the  AP coefficient does not change much across different wealth levels. 
Recall that the exponential utility has a constant  AP coefficient, and it satisfies \eqref{eq:suff}. Indeed, the shortfall risk with an exponential utility, with the corresponding loss function given by
$$\ell(x)= e^{\gamma x}-1,  ~~x\in \R,~\gamma>0,$$
also belongs to the class of CE studied  in Section \ref{sec:CE}, and we know it is submodular from Proposition \ref{prop:CE}; it  also follows from Theorem \ref{thm:linear-iff}.
In the example below, we see that \eqref{eq:LD} allows for a larger class of shortfall risk measures.


\begin{example} 
Let $\ell(x)=e^{2x}+e^{x}-2$. Then
$$
R(x)=1+\frac{2e^x}{2e^x+1},\qquad x\in\R,
$$
and we can see that $R$ lies strictly between $1$ and $2$. Therefore, \eqref{eq:suff} holds, and the corresponding shortfall risk measure $\rho_\ell$ is submodular. 
\end{example}

 \begin{example}
Take $\ell(x)=x$. The corresponding shortfall risk is the mean. Note that $R(x)=0$ for all $x\in \R$ in this example, and thus condition \eqref{eq:suff} holds. Thus, it is submodular by Theorem \ref{thm:linear-iff}.  Indeed,  we know the mean  is modular by Theorem \ref{th:modular-char}. 
\end{example}
 \begin{example}
The loss function for a convex expectile is given by
$$
\ell(x)=x+a \max\{x,0\},\qquad x\in \R,~a\ge 0.
$$
By Theorem \ref{th:distortion}, a convex expectile cannot be submodular (except for the case $a=0$, corresponding to the mean), because it is coherent but not a distortion risk measure (Theorem \ref{th:distortion} states that submodular and coherent risk measures must be distortion risk measures).  
Note that for $a\ne 0$, $\ell$ is not differentiable at $0$, so we cannot directly apply Theorem \ref{thm:linear-iff}; we need the more general result stated in Appendix \ref{app:proofs}, which does not require differentiability. However, the intuition still applies: interpreting $R$ as the curvature (as a limit), then $R(0)=\infty$
and $R(x)=0$ for $x\ne 0$. Thus, with this limiting interpretation, \eqref{eq:LD} cannot hold, as its left-hand side is infinite, but the right-hand side is finite. 
\end{example}

\subsection{Optimized certainty equivalents}
 \citet{BT07} studied the optimized certainty equivalent (OCE) as risk measures.  For a convex function $\ell$, the OCE is defined as   
$$
\mathrm{OCE}_\ell(X)=\inf_{m\in\R}\big\{m+\E[\ell(X-m)]\big\},\qquad X\in L^\infty.
$$
It is known that OCEs are convex risk measures.

\begin{theorem}\label{thm:oce-submod}
Assume $\ell:\R\to\R$ is increasing and convex, and $\mathrm{OCE}_\ell(X)\in\R$ for all $X\in L^\infty$. Then the loss-based optimized certainty equivalent $\mathrm{OCE}_\ell$ is submodular.
\end{theorem}

\begin{proof}
Fix $X,Y\in L^\infty$ and $\varepsilon>0$. Choose $m_X,m_Y\in\R$ with
$$
m_X+\E[\ell(X-m_X)]\le \mathrm{OCE}_\ell(X)+\varepsilon,\qquad
m_Y+\E[\ell(Y-m_Y)]\le \mathrm{OCE}_\ell(Y)+\varepsilon.
$$
Without loss of generality, assume $m_X\ge m_Y$. Since $m_X$ and $m_Y$ are not necessarily optimal for $X\vee Y$ and $X\wedge Y$, the definition of $\mathrm{OCE}_\ell$ gives
$$
\mathrm{OCE}_\ell(X\vee Y)+\mathrm{OCE}_\ell(X\wedge Y)
\le m_X+m_Y+\E\!\big[\ell\big((X\vee Y)-m_X\big)+\ell\big((X\wedge Y)-m_Y\big)\big].
$$
It therefore suffices to show that, pointwise,
\begin{equation}\label{eq:pw}
\ell\big(a\vee b -m_X\big)+\ell\big(a\wedge b -m_Y\big)
\le \ell(a-m_X)+\ell(b-m_Y),
\end{equation}
where $a=X(\omega)$ and $b=Y(\omega)$. If $a\ge b$, both sides are equal. If $a<b$, set $s=b-a>0$. Then $(a\vee b)-m_X=(a-m_X)+s$ and $(a\wedge b)-m_Y=(b-m_Y)-s$. Since $m_X\ge m_Y$, we have $a-m_X\le b-m_Y-s$, and because the increment $x\mapsto\ell(x+s)-\ell(x)$ is increasing for convex $\ell$,
$$
\ell(a-m_X+s)-\ell(a-m_X)\le \ell(b-m_Y)-\ell(b-m_Y-s),
$$
which rearranges to \eqref{eq:pw}. Integrating over $\omega$ gives
$$
\mathrm{OCE}_\ell(X\vee Y)+\mathrm{OCE}_\ell(X\wedge Y)\le \mathrm{OCE}_\ell(X)+\mathrm{OCE}_\ell(Y)+2\varepsilon.
$$
Since $\varepsilon>0$ is arbitrary, the result follows.
\end{proof}

In contrast to the other families studied here, the OCE family requires no extra structural restriction for submodularity.

\subsection{Adjusted ES}
 \citet{BMW22} defined the AES as
$$
\ES^g(X)=\sup_{p\in[0,1]}\{\ES_p(X)-g(p)\},\qquad X\in L^\infty.
$$
where $g:[0,1]\to[0,\infty]$ is increasing with $g(0)<\infty$. 
Since submodularity is invariant under constant shift, we can without loss of generality assume $\ES^g(0)=0$, that is, $g(0)=0$.

\begin{theorem}\label{thm:AES-submodular-characterization}
Let $g:[0,1]\to[0,\infty]$ be increasing with $g(0)=0$. Then  $\ES^g$ is submodular if and only if there exists
$p_0\in[0,1]$ such that
$
\ES^g =\ES_{p_0} . 
 $
\end{theorem}

\begin{proof}
The ``if'' statement is straightforward because  ES is
submodular.
For the ``only if'' statement, assume that $\rho$ is submodular, and write $\rho=\ES^g$ for simplicity.
We first treat the case $g(1-)=\infty$. The remaining case 
$g(1-)<\infty$ is quite involved and is treated in Section \ref{app:AES-proof} of Appendix \ref{app:proofs}. 

Let $
p_0=\sup\{p\in[0,1]:g(p)=0\}.
 $
Since $g$ is increasing and $g(1-)=\infty$, we have $p_0<1$, and
$g(p)>0$ for every $p>p_0$.

Fix $\lambda>0$. For $t\in[0,1]$, define
$$
\phi_\lambda(t)=\rho(\lambda\id_E)\qquad\text{whenever }\p(E)=t.
$$
The mapping
$\phi_\lambda$ is well defined by atomlessness and law invariance. Indeed,
submodularity of $\rho$ implies that the set function
$E\mapsto \rho(\lambda\id_E)$ is submodular on $(\mathcal F,\cup,\cap)$, because
$\id_{A\cup B}=\id_A\vee\id_B$ and $\id_{A\cap B}=\id_A\wedge\id_B$. Since
$E\mapsto \rho(\lambda\id_E)$ depends only on $\p(E)$, we have 
$$
\phi_\lambda(s+h)-\phi_\lambda(s)\ge \phi_\lambda(t+h)-\phi_\lambda(t),
\qquad 0\le s\le t\le t+h\le 1.
$$
Hence $\phi_\lambda$ is concave on $[0,1]$. Moreover,
\begin{equation}\label{eq:AES-phi-infinite}
\phi_\lambda(t)
=
\sup_{p\in[0,1]}
\left\{
\lambda\min\!\left(1,\frac{t}{1-p}\right)-g(p)
\right\},
\qquad t\in[0,1].
\end{equation} 
We claim 
\begin{equation}\label{eq:AES-origin-slope-infinite}
\lim_{t\downarrow0}\frac{\phi_\lambda(t)}{t}
=
\frac{\lambda}{1-p_0}.
\end{equation}

\underline{Lower bound.}
Let $r\in[0,1)$ satisfy $g(r)=0$. Then, for $0<t<1-r$, using \eqref{eq:AES-phi-infinite}, 
$$
\phi_\lambda(t)\ge \frac{\lambda t}{1-r}.
$$
Hence
$$
\liminf_{t\downarrow0}\frac{\phi_\lambda(t)}{t}
\ge
\sup\left\{\frac{\lambda}{1-r}:r\in[0,1),\ g(r)=0\right\}
=
\frac{\lambda}{1-p_0}.
$$

\underline{Upper bound.}
Suppose to the contrary that there exist $\varepsilon>0$ and $t_n\downarrow0$ such that
$$
\frac{\phi_\lambda(t_n)}{t_n}\ge \frac{\lambda}{1-p_0}+2\varepsilon,
\qquad n\in\mathbb N.
$$
For each $n$,    using \eqref{eq:AES-phi-infinite} and $g(1-)=\infty$, we can find  $p_n\in[0,1)$ such that
\begin{equation}
    \label{eq:AES-pf-1}
\lambda\min\!\left(1,\frac{t_n}{1-p_n}\right)-g(p_n)
\ge \phi_\lambda(t_n)-\varepsilon t_n
\ge \left(\frac{\lambda}{1-p_0}+\varepsilon\right)t_n>0,
\end{equation} 
so $g(p_n)<\lambda$. Since $g(1-)=\infty$, there exists $\bar p<1$ such that
$g(p)>\lambda$ for all $p>\bar p$. Thus $p_n\le \bar p$ for all $n$.
For    $n$ large enough, $t_n<1-\bar p$. Hence $t_n<1-p_n$, and \eqref{eq:AES-pf-1} yields
$$
\frac{\lambda}{1-p_n}-\frac{g(p_n)}{t_n}
\ge
\frac{\lambda}{1-p_0}+\varepsilon.
$$
In particular,
$$
\frac{\lambda}{1-p_n}\ge \frac{\lambda}{1-p_0}+\varepsilon.
$$
Since $p\mapsto \lambda/(1-p)$ is continuous and increasing on $[0,\bar p]$, there exists
$\delta>0$ such that $p_n\ge p_0+\delta$ for all large $n$. Hence
$g(p_n)\ge g(p_0+\delta)>0$, which implies
$$
\frac{\lambda}{1-p_n}-\frac{g(p_n)}{t_n}\to -\infty,
$$
a contradiction. The two bounds together prove \eqref{eq:AES-origin-slope-infinite}.

Now suppose there exists $p_1>p_0$ such that $g(p_1)<\infty$. Choose
$\lambda$   large enough such that
\begin{equation}\label{eq:AES-lambda-infinite}
\lambda>\frac{g(p_1)(1-p_0)}{p_1-p_0}.
\end{equation}
Set $t_1=1-p_1$. Using $p=p_1$ in \eqref{eq:AES-phi-infinite}, we get
$\phi_\lambda(t_1)\ge \lambda-g(p_1)$. Therefore,
$$
\frac{\phi_\lambda(t_1)}{t_1}
\ge
\frac{\lambda-g(p_1)}{1-p_1}
>
\frac{\lambda}{1-p_0}.
$$
But $\phi_\lambda$ is concave and $\phi_\lambda(0)=0$, so the map
$t\mapsto \phi_\lambda(t)/t$ is nonincreasing on $(0,1]$, contradicting
\eqref{eq:AES-origin-slope-infinite}. Hence $g(p)=\infty$ for every $p>p_0$. 
Consequently, for every $X\in L^\infty$,
$$
\rho(X)=\sup\{\ES_p(X):p\in[0,1],\ g(p)=0\}=\ES_{p_0}(X).
$$
This proves the claim.
\end{proof}

The intuition is that $\ES^g$ is the upper envelope of the family
$\{\ES_p-g(p):p\in[0,1]\}$. Although each $\ES_p$ is submodular, taking the pointwise supremum over many levels typically introduces switching between different active ES levels. Such switching is incompatible with the concavity structure imposed by submodularity on indicator-based profiles. Hence submodularity can persist only when the envelope does not switch across levels, that is, when $\ES^g$ collapses to a single $\ES_{p_0}$.

\subsection{Monotone mean-deviation risk measures}
A monotone mean-deviation risk measure, studied by \citet{HWW26}, is defined as  
\begin{equation}\label{eq:MD-1}
\rho(X) = g(\rho_\phi(X)-\E[X]) + \E[X] ,~~~X\in L^\infty,
\end{equation}
where $g:[0,\infty)\to [0,\infty)$ is an increasing, convex, and non-constant function satisfying $g(0)=0$,
and $\rho_\phi:L^\infty\to\R$ is a distortion risk measure with a concave distortion function $\phi$.
In \citet{HWW26}, $\rho_\phi$ in \eqref{eq:MD-1} can be replaced by other law-invariant coherent risk measures, but for explicit formulas, we stick to the setting of using distortion risk measures in \eqref{eq:MD-1}.

\begin{theorem}\label{th:md}
The risk measure $\rho$ in \eqref{eq:MD-1} is submodular if and only if it is a coherent distortion risk measure, that is, $g$ is linear or $\rho_\phi(X)=\E[X]$ for all $X\in L^\infty$.
\end{theorem}

\begin{proof} The ``if'' statement is straightforward from Theorem \ref{th:distortion}. Below we show the ``only if'' statement.
Suppose that the convex function $g$ is not linear and $\rho_\phi(X)\neq \E[X]$ for some $X\in L^\infty$, and we will show that $\rho$ is not submodular. There exists $x>0$ such that $g$ is 
locally nonlinear at $x$, and by convexity it implies $g(y)+g(z)>2g(x)$ for all $y,z\ge 0$ with $y+z=2x$ and $y\ne x$. 

Since $\phi$ is not the identity and it is concave, we know $\phi(t)>t$ for all $t\in (0,1)$. 
Let $\psi(t)=\phi(t)-t$ for $t\in (0,1)$.
The set $\{\psi(t) : t\in (0,1)\}$ is a nonempty interval because of the continuity of $\psi$ on $(0,1)$.
Moreover, $\psi(t) \to 0$ as $t\uparrow 1$. Therefore, 
we can find $p,q,r\in (0,1)$ with $p<q<r$ such that 
$\psi(p) > \psi(q)  >  \psi(r) $ and  $\psi(p)+\psi(r) = 2 \psi(q)$. 
 Let $m= x/\psi(q)$, 
$X=m (\id_A+\id_C)/2$, and $Y=  m \id_B $, where the events $A,B,C$ satisfy $A\subseteq B\subseteq C$,  $\p(A)=p$, $\p(B)=q$, and $\p(C)=r$.
We can calculate 
\begin{align*}  \rho_\phi(X)-\E[X] &= m (\psi(p) + \psi(r) )/2= x  ;\\
\rho_\phi(Y)-\E[Y]&= m \psi(q) =x;  \\
 y= \rho_\phi( X \wedge Y)  -\E[X\wedge Y] &=m ( \psi(p)+\psi(q) )/2 >x ;\\
  z= \rho_\phi( X \vee Y)  -\E[X\vee  Y] &  = m (  \psi(q)+\psi(r) )/2<x.
\end{align*}
Since $y+z=2x$ and $y>x$, we get
$$
\rho(X \wedge Y) +  \rho(X \vee Y) 
= g(y)+ g(z) +\E[X+Y] > 2g(x)+\E[X+Y] = \rho(X)+\rho(Y),
$$
showing that $\rho$ is not submodular.
\end{proof}

Theorem \ref{th:md} shows that submodularity is highly restrictive in this class. Submodularity excludes any active nonlinearity in the deviation weighting: either $g$ is linear, or $\rho_\phi(X)=\E[X]$ for all $X\in L^\infty$, so the deviation term vanishes.

\section{Submodularity in financial data}
\label{sec:7}

In this section, we illustrate empirical observations of submodularity and non-submodularity of risk measures in financial data.

\paragraph{Methodology.}
For each analysis, we compute daily log returns as
$r_t=\ln(P_t)-\ln(P_{t-1}),$ where $P_t$ denotes the adjusted closing price on day $t$.
Within each rolling window, returns are converted to losses $L_t=-r_t$. For each stock and each trading day, we compute rolling historical VaR and ES using two window lengths:  $n=250$ trading days (approximately one year) and $n=500$ trading days (approximately two years).

For a confidence level $p\in(0,1)$, let $k=\lceil n(1-p)\rceil,$ and let $L_{(1)}\ge L_{(2)}\ge \cdots \ge L_{(n)}$ denote the ordered losses in descending order (largest loss first). We compute:
$$
\VaR_{p}=L_{(k)},\qquad
\ES_{p}=\frac{1}{k}\sum_{i=1}^{k}L_{(i)}.
$$
That is, we use a rolling historical estimator based on upper-tail order statistics: $\VaR_p$ is the $k$-th largest loss in the window, and $\ES_p$ is the arithmetic mean of the $k$ largest losses. This convention is immaterial under continuous loss distributions and is conservative in finite samples when $n(1-p)$ is an integer.
Throughout this section, $p$ denotes the confidence level, so the corresponding tail probability is $1-p$.
For the sector-based analysis, we use $p\in\{0.90,\,0.95\}$.
The pair-based analysis explores 14 levels:
$p\in\{0.99,\,0.98,\,\ldots,\,0.90\}$ and $p\in\{0.88,\,0.85,\,0.82,\,0.80\}$,
to examine non-monotonic patterns across confidence levels.
Following the theoretical framework, we consider the ES-based two-level specification
$$
\AES_{p,q,c}(X)=\max\bigl\{\ES_q(X),\ \ES_p(X)-c\bigr\},
$$
where $p=0.98$ and $c$ is a fixed constant. This specification is motivated by the general AES form
$$
\sup_{u\in[0,1]}\{\ES_u(X)-g(u)\},
$$
and is used here as a tractable empirical specification.
For the sector-based analysis, $q\in\{0.90,0.95\}$.
For the pair-based analysis, $q$ uses the same 14 levels as $p$.
We take $c\in\{0.01,0.015,0.02\}$.
The choice of $c$ is calibrated to the empirical gap between $\ES_p$ and $\ES_q$. In our sample, when $q=0.90$, the mean gap $\ES_{0.98}-\ES_{0.90}$ ranges from 0.023 to 0.040 across stocks and window lengths. Thus, these values of $c$ are of the same order of magnitude and have a nontrivial effect on $\AES_{p,q,c}$.

\paragraph{Submodularity test.}
For each pair of stocks $(X,Y)$ and each trading day, we test the submodularity condition, where $\rho$ denotes the risk measure (VaR, ES, or $\AES_{p,q,c}$).
We record a violation when
$$
\rho(X)+\rho(Y)-\rho(X\wedge Y)-\rho(X\vee Y)<-\epsilon,
\qquad \epsilon=10^{-8},
$$
where $\epsilon=10^{-8}$ is a conservative tolerance used to guard against floating-point error. Here $X\wedge Y$ and $X\vee Y$ are implemented at the loss-series level: within each rolling window of length $n$, we form $L_t^{X\wedge Y}=\min\{L_t^X,L_t^Y\}$ and $L_t^{X\vee Y}=\max\{L_t^X,L_t^Y\},$ then apply the same rolling historical estimator to these constructed series.
The daily violation rate is computed as the proportion of stock pairs exhibiting a violation on each day.

\paragraph{Sample construction and analysis designs.}
We conduct two complementary empirical analyses that differ in sample construction and aggregation. To examine how submodularity violations vary across confidence levels, we perform an analysis on the three selected stock pairs (META--NFLX, DIS--GOOGL, and DIS--META), chosen to represent high-volatility technology and media firms with strong tail-risk interactions. For the pair-based analysis, daily adjusted closing prices are obtained from Stooq. For the sector-based panel and the VIX series, data are downloaded from Yahoo Finance. The download window begins approximately 750 trading days before January 1, 2018 in order to provide a warm-up period for the rolling-window estimators. For each pair, returns are aligned on common trading dates and rows with any missing observations are dropped (complete-case deletion); no forward-filling or interpolation is applied. The reported pair-based results are restricted to January 1, 2018 onward, yielding 1{,}781 valid trading days for each of the 250-day and 500-day windows. The analysis focuses on how submodularity violations vary with confidence levels, window lengths, and (for $\AES_{p,q,c}$) the parameters $(q,c)$.

The second part of the analysis is a sector-based design intended to provide broad market coverage while remaining computationally tractable. Specifically, we use the latest S\&P~500 constituent snapshot before January 2, 2018 and select up to five constituents by market capitalization from each historical GICS sector, resulting in 52 unique stocks (see Table~\ref{tab:stock_selection} in the Appendix for the full list with market capitalizations and data sources). The historical sector classification predates the later Communication Services reclassification, so the Telecommunication Services sector contributes only two stocks in the snapshot. Using beginning-of-sample market capitalizations avoids look-ahead bias in stock selection. The sector-based results, figures, and exported statistics cover the period from January 2, 2018 to January 30, 2025 (1{,}780 trading days). For the submodularity test, we evaluate all pairwise combinations of these 52 stocks, yielding $\binom{52}{2}=1{,}326$ pairs per trading day.

\paragraph{Pair-based results.}
For the pair-based analysis, Figure~\ref{fig:es_var_conf} displays VaR and ES violation rates as a function of confidence level for each pair. The VaR violation rate shows a non-monotonic pattern that varies across window lengths. Averaged across the three pairs, the 250-day window gives violation rates of 15.3\% at 80\% confidence, 8.7\% at 95\% confidence, and 36.1\% at 99\% confidence. For the 500-day window, the rates are 0.9\% at 95\% confidence, 26.3\% at 98\% confidence, and 22.5\% at 99\% confidence. Thus the longer window greatly reduces violations at moderate confidence levels, but the extreme-tail behavior remains irregular. As expected from the exact ES structure of the estimator, ES exhibits no submodularity violations across all confidence levels and both rolling-window lengths in our sample (Table~\ref{tab:pair_violation}).

Figure~\ref{fig:esg_conf} illustrates how $\AES_{p,q,c}$ violation rates vary with the choice of $q$ and $c$. Two patterns emerge. First, the effect of $q$ is non-monotone: violation rates are highest for intermediate values of $q$ and are close to zero at both extremes. When $q$ is very close to $p$, the adjustment term dominates and $\AES_{p,q,c} = \ES_p - c$, which preserves submodularity; when $q$ is far from $p$, $\AES_{p,q,c} = \ES_q$, which is also submodular. Second, larger values of $c$ generally produce higher violation rates. For example, at $q=0.94$, the mean violation rate rises from 0.00\% for $c=0.01$ to 1.28\% for $c=0.015$ and 6.01\% for $c=0.02$. Overall, $\AES_{p,q,c}$ exhibits an average violation rate of 1.12\%, placing it between ES (no observed violations) and VaR (frequent violations) (Table~\ref{tab:pair_violation}).

\begin{table}[htbp]
\centering
\caption{Pair-based submodularity violation rates}
\label{tab:pair_violation}
\begin{tabular}{lrrr}
\hline
Risk measure & Total tests & Violations & Rate (\%) \\
\hline
VaR          & 149{,}604 & 15{,}398 & 10.29 \\
ES           & 149{,}604 & 0        & 0.00 \\
$\AES_{p,q,c}$ & 448{,}812 & 5{,}019  & 1.12 \\
\hline
\multicolumn{4}{@{}l@{}}{\scriptsize 3 pairs; both windows have $N=1{,}781$; 14 confidence levels; $\AES_{p,q,c}$ uses 14 $q$-levels and 3 $c$-values.} \\
\end{tabular}
\end{table}

\begin{figure}[p]
  \centering
  \includegraphics[width=\textwidth,height=0.24\textheight,keepaspectratio]{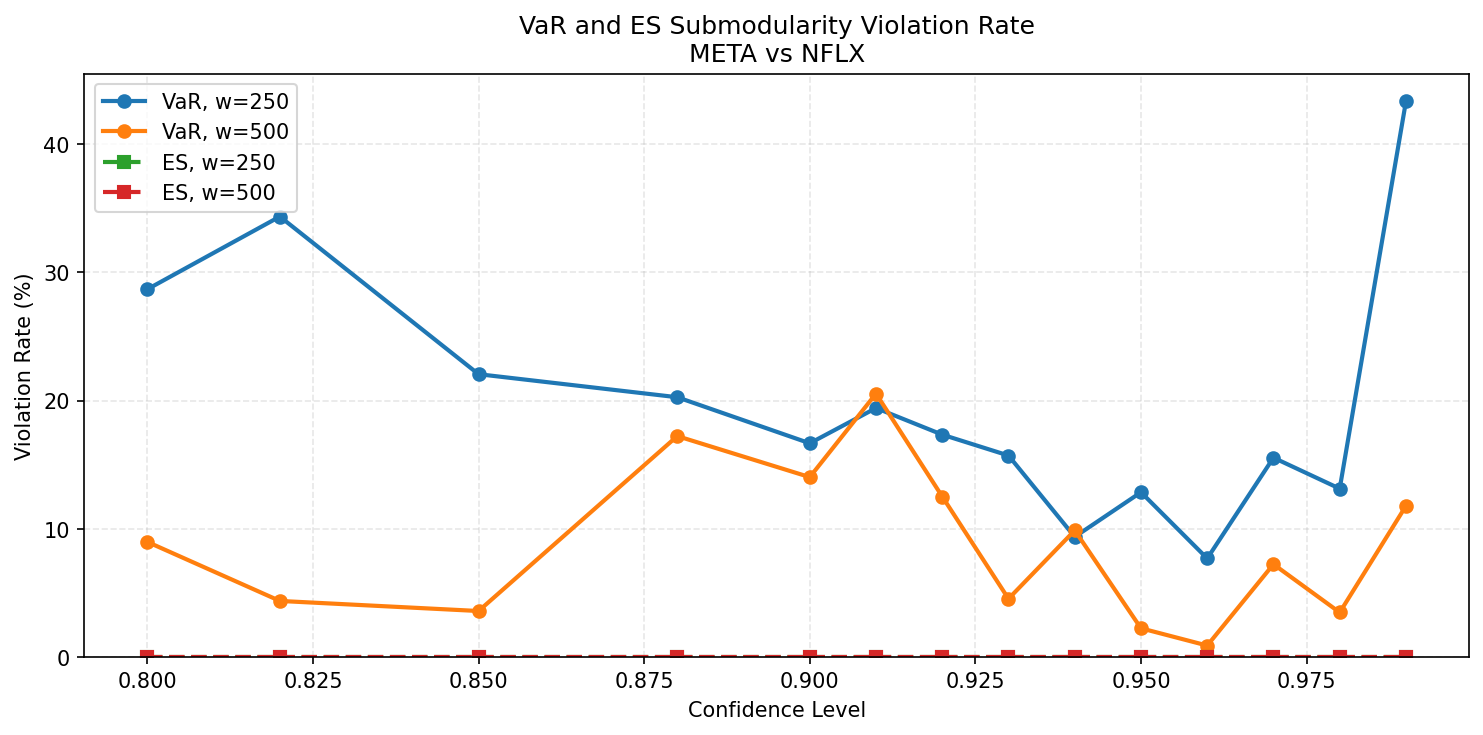}
  \par\vspace{0.3em}
  \includegraphics[width=\textwidth,height=0.24\textheight,keepaspectratio]{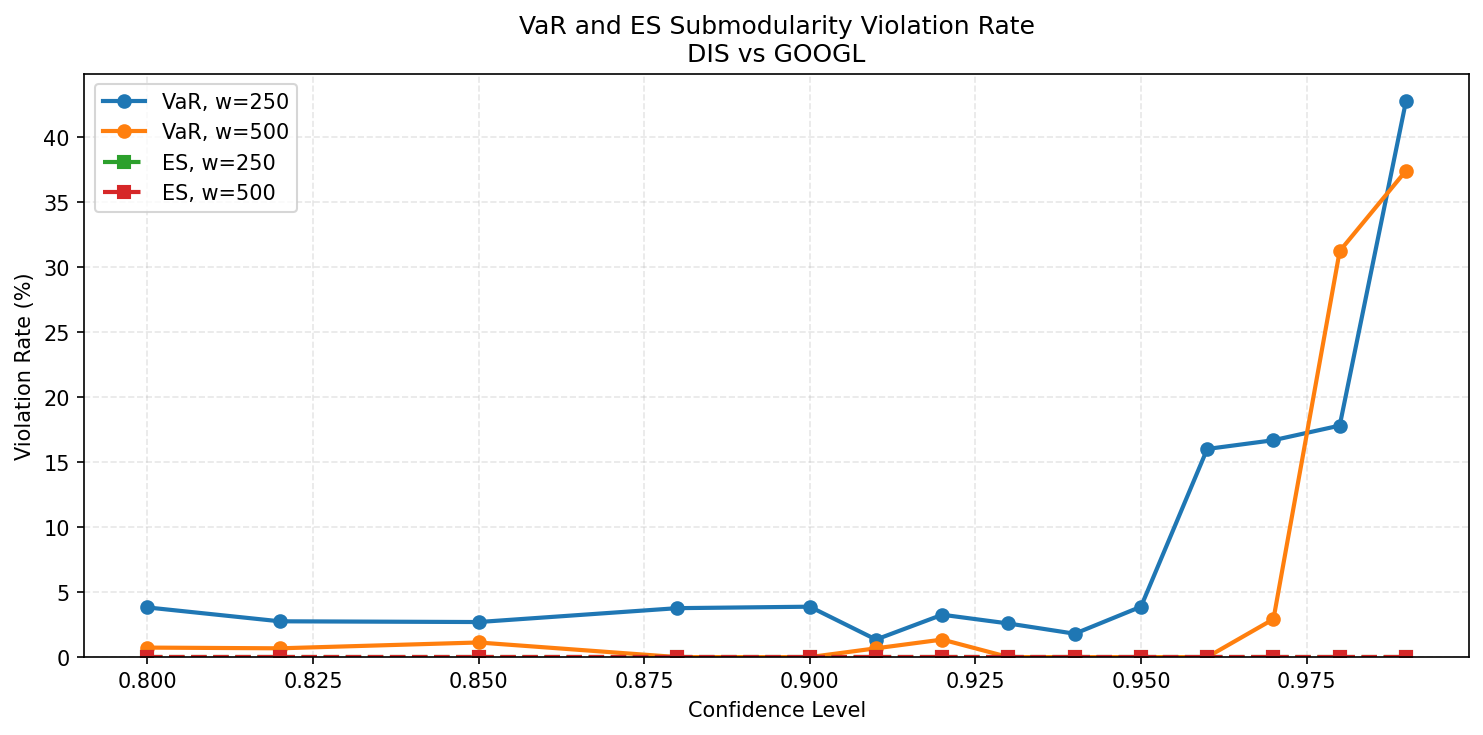}
  \par\vspace{0.3em}
  \includegraphics[width=\textwidth,height=0.24\textheight,keepaspectratio]{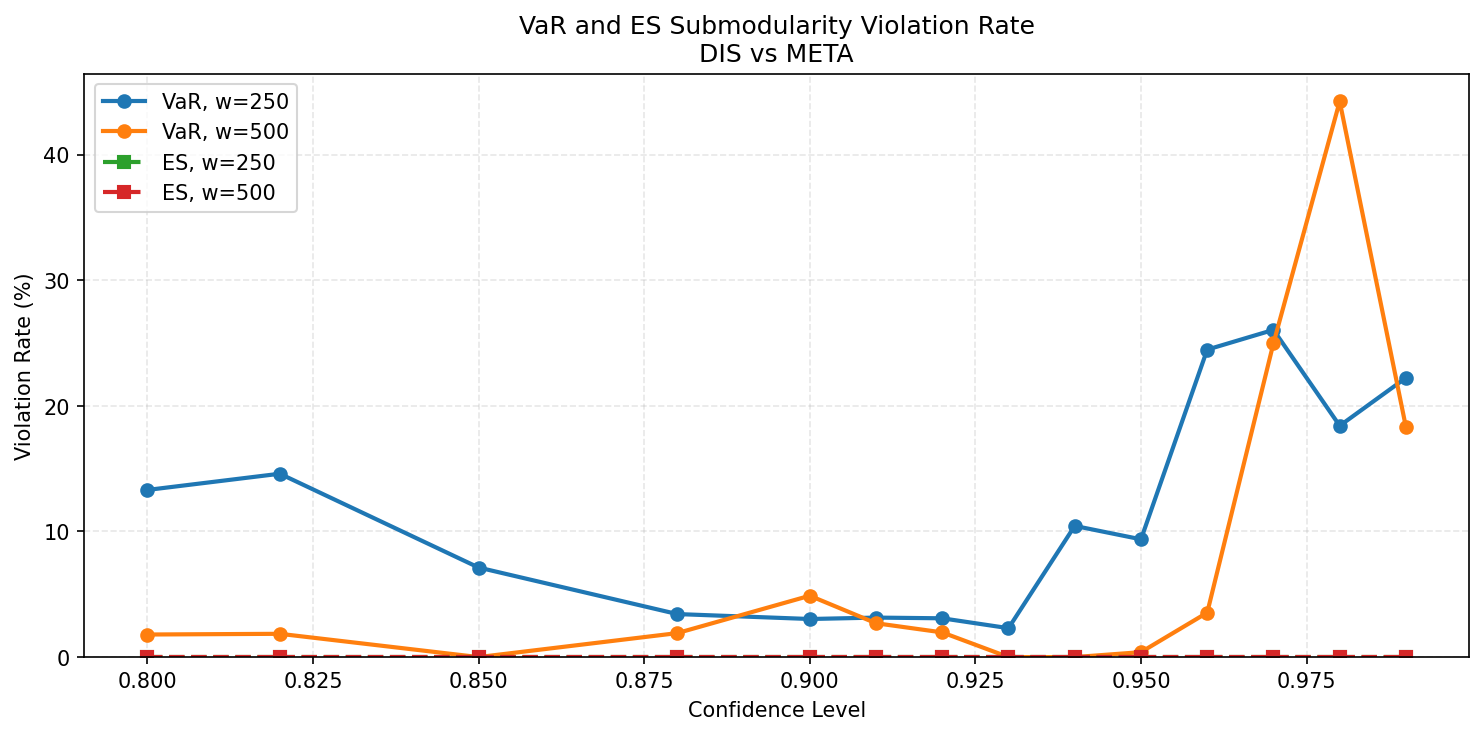}
  \caption{VaR and ES: violation rate by confidence level for the three pairs, ordered top to bottom as META--NFLX, DIS--GOOGL, and DIS--META.}
  \label{fig:es_var_conf}
\end{figure}

\begin{figure}[p]
  \centering
  \includegraphics[width=\textwidth,height=0.24\textheight,keepaspectratio]{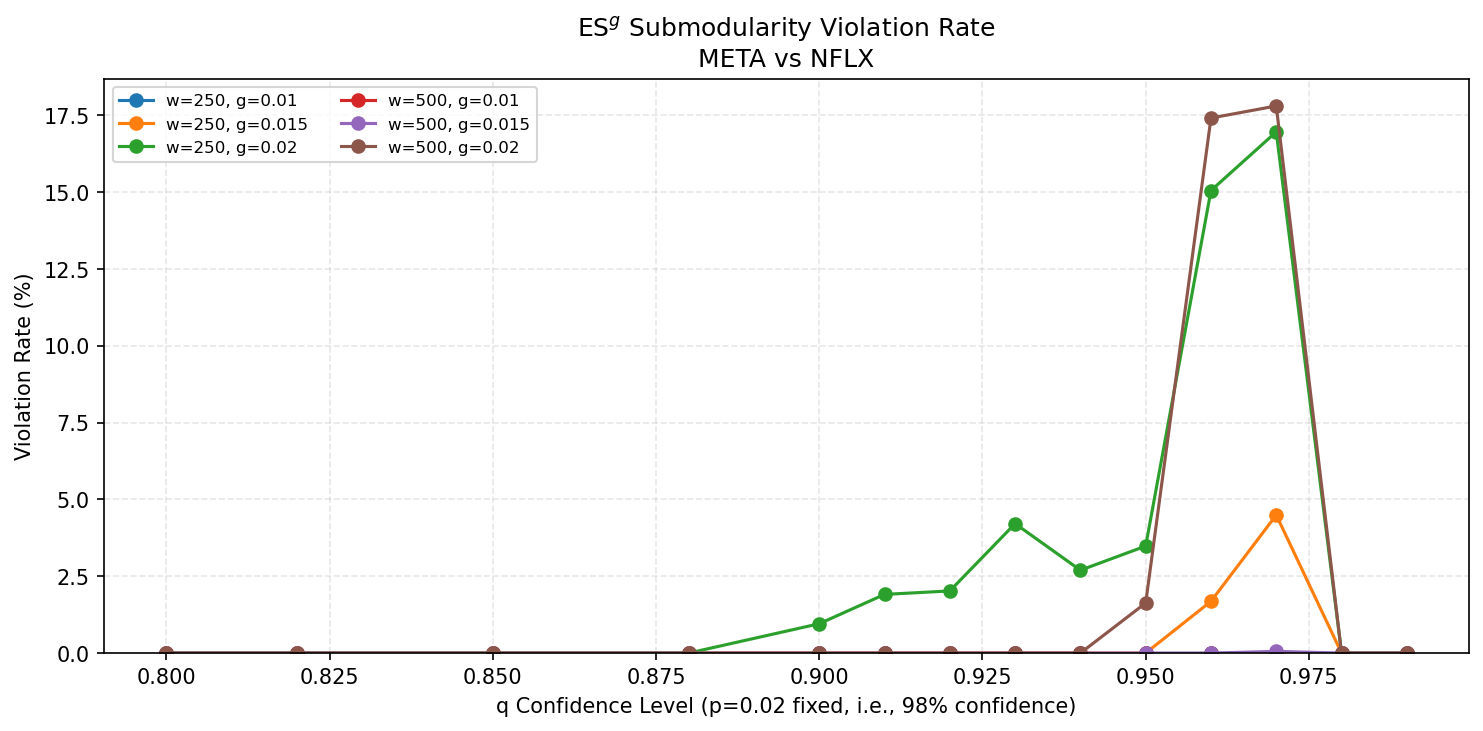}
  \par\vspace{0.3em}
  \includegraphics[width=\textwidth,height=0.24\textheight,keepaspectratio]{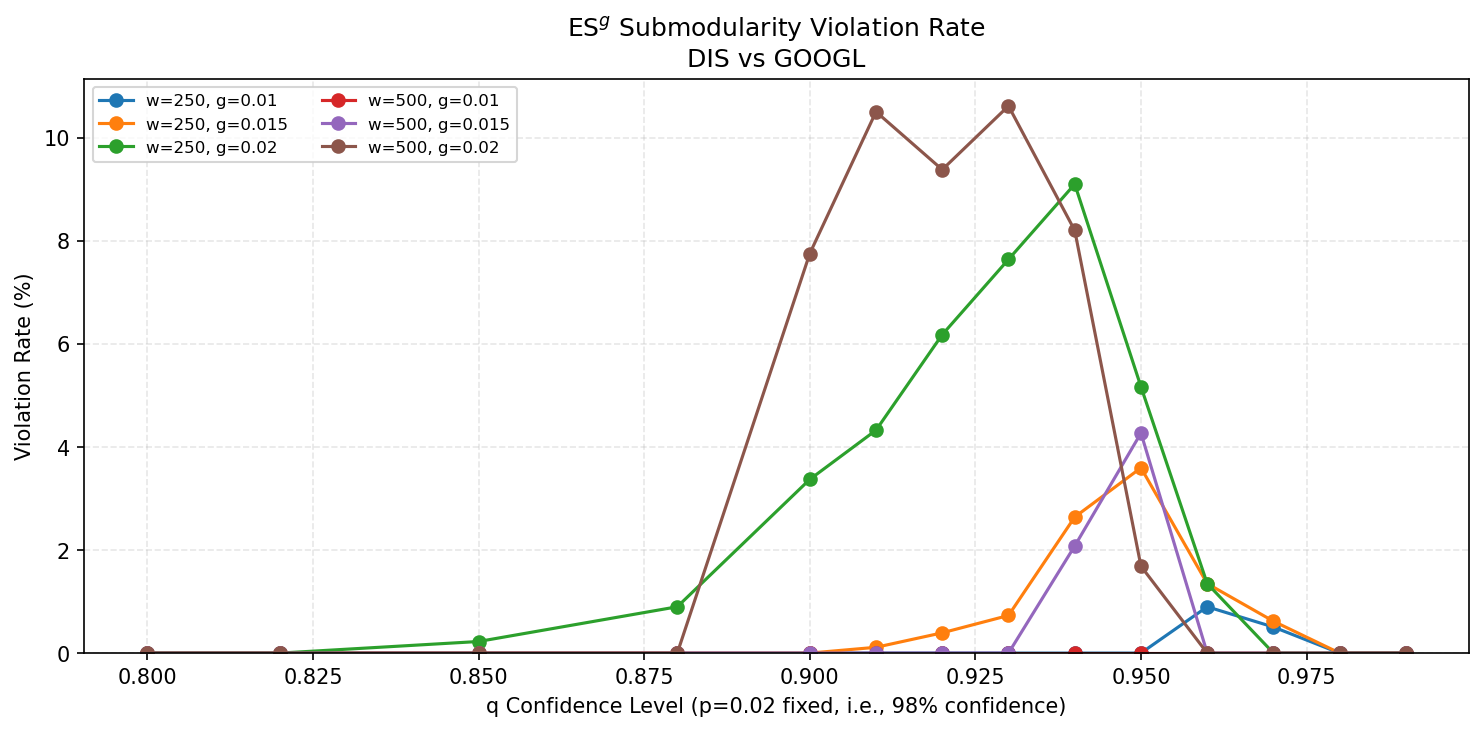}
  \par\vspace{0.3em}
  \includegraphics[width=\textwidth,height=0.24\textheight,keepaspectratio]{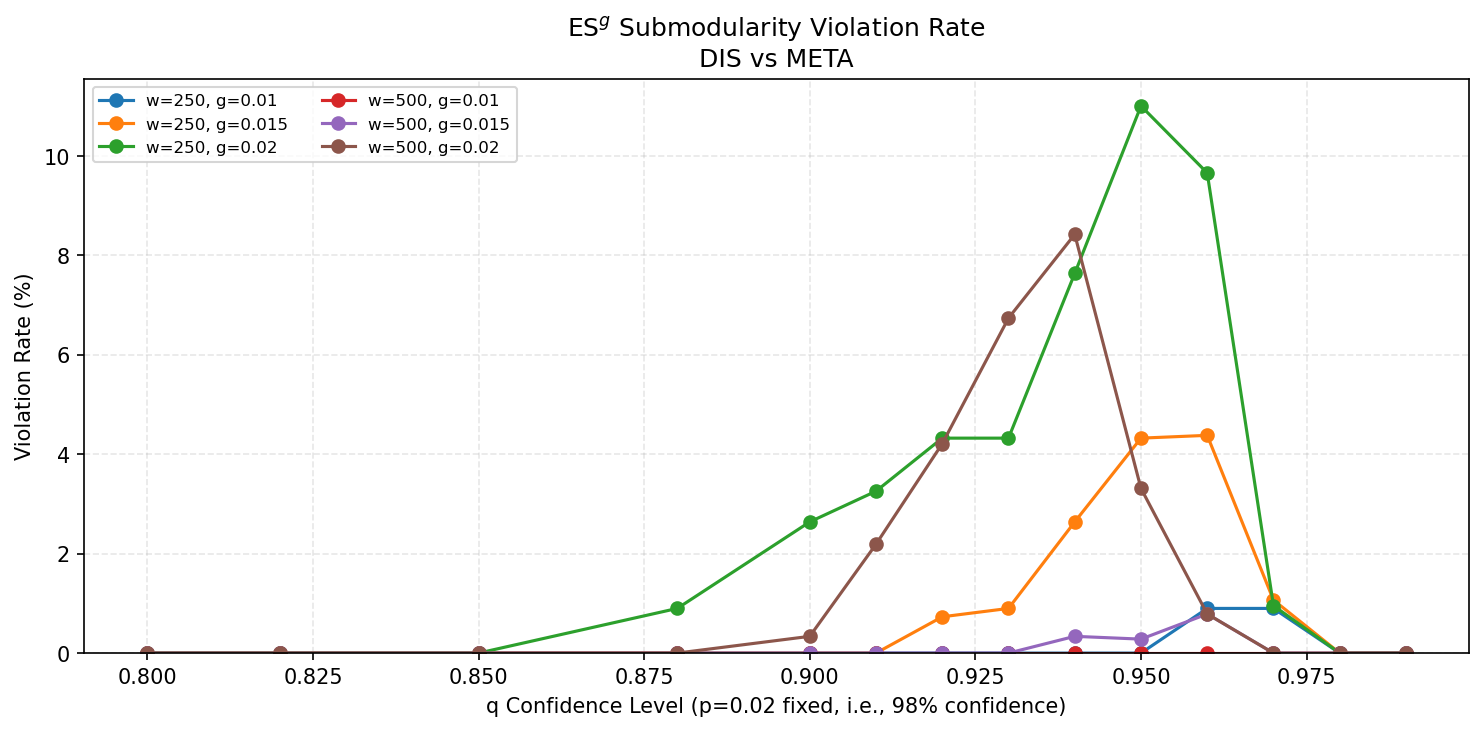}
  \caption{$\AES_{p,q,c}$: violation rate by $q$ for $c=0.01,0.015,0.02$, with pairs ordered top to bottom as META--NFLX, DIS--GOOGL, and DIS--META.}
  \label{fig:esg_conf}
\end{figure}

\paragraph{Sector-based results.}
For the sector analysis, Table~\ref{tab:sector_violation} reports summary statistics of the daily submodularity violation rates, aggregated across all stock pairs and trading days, for each risk measure. Overall, the empirical evidence is consistent with the theoretical predictions. As expected from the exact ES structure of the estimator, no ES violations are observed over the entire sample period. In contrast, VaR frequently violates submodularity: the average daily violation rate is about 4.5\%, with standard deviation 3.8\% and a peak of 24.7\%, which tends to occur during periods of elevated market volatility. Longer estimation windows substantially reduce mean VaR violation rates, from 6.4\% at $n=250$ to 2.7\% at $n=500$, suggesting that sampling variability in the quantile estimate is a key driver of submodularity failures. $\AES_{p,q,c}$ displays smaller but still visible violations, with mean 0.40\% and maximum 5.05\%, again placing it between ES (no observed violations) and VaR (frequent violations). The adjustment parameter $c$ remains important: for the 250-day window with $q=0.90$, the mean violation rate increases from 0.10\% for $c=0.01$ to 0.64\% for $c=0.015$ and 1.56\% for $c=0.02$. To illustrate the time-series behavior, Figures~\ref{fig:var_es_violations} and~\ref{fig:esg_violations} plot the daily violation rates for VaR/ES and $\AES_{p,q,c}$, respectively.

\begin{table}[htbp]
\centering
\caption{Sector-based daily submodularity violation rates}
\label{tab:sector_violation}
\small
\begin{tabular}{lcccc}
\hline
Risk measure & Mean (\%) & Std (\%) & Min (\%) & Max (\%) \\
\hline
VaR          & 4.52 & 3.79 & 0.30 & 24.74 \\
ES           & 0.00 & 0.00 & 0.00 & 0.00 \\
$\AES_{p,q,c}$ & 0.40 & 0.70 & 0.00 & 5.05 \\
\hline
\multicolumn{5}{@{}l@{}}{\scriptsize 2018-01-02 to 2025-01-30; VaR/ES at 90\%, 95\%; $\AES_{p,q,c}$ with $q=0.90,0.95$ and $c=0.01,0.015,0.02$.} \\
\end{tabular}
\end{table}

\begin{figure}[htbp]
  \centering
  \includegraphics[width=\textwidth]{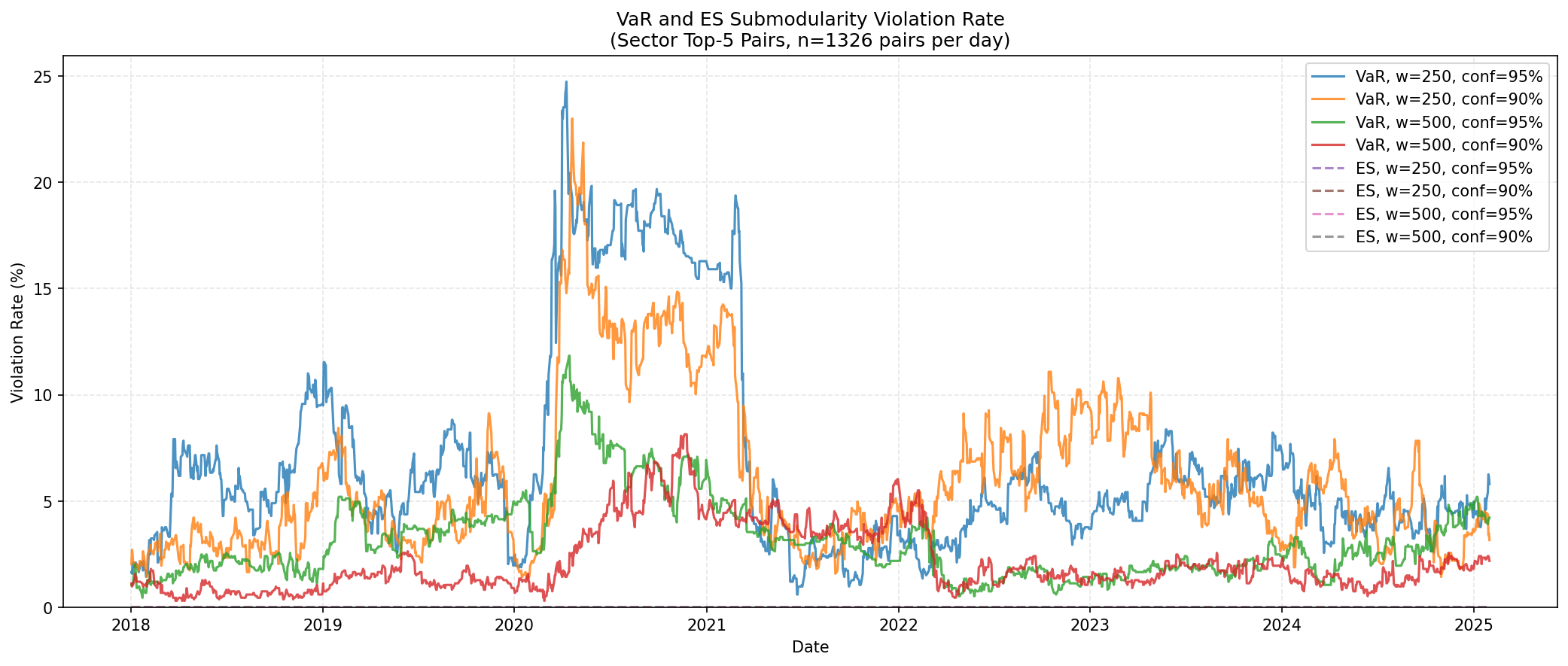}
  \caption{Daily VaR and ES submodularity violation rate.}
  \label{fig:var_es_violations}
\end{figure}

\begin{figure}[htbp]
  \centering
  \includegraphics[width=\textwidth]{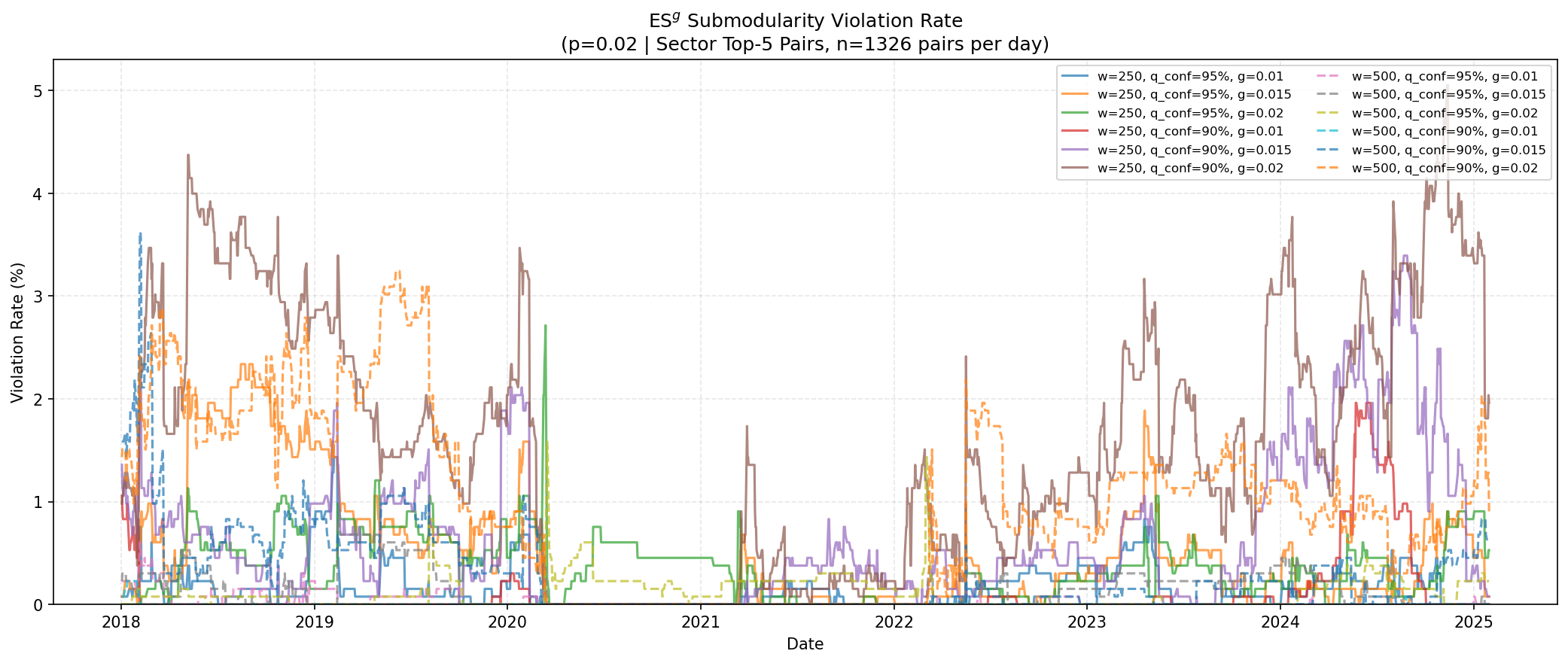}
  \caption{Daily $\AES_{p,q,c}$ submodularity violation rate.}
  \label{fig:esg_violations}
\end{figure}

\paragraph{Subadditivity test and VIX correlation analysis.}
As a complementary analysis, we test VaR for subadditivity violations using the same sector-based sample of stock pairs. For each pair $(X,Y)$, we form the portfolio loss series $L_t^{X+Y}=-\bigl(r_t^X+r_t^Y\bigr)$ and compute the rolling VaR of the sum directly. A subadditivity violation is recorded when
$$
\VaR_p(X)+\VaR_p(Y)-\VaR_p(X+Y)<-\epsilon,
\qquad \epsilon=10^{-8},
$$
that is, when $\VaR_p(X+Y)>\VaR_p(X)+\VaR_p(Y)$, so that VaR penalizes diversification. We test at confidence levels $p\in\{0.95,\,0.975,\,0.99\}$ using both window lengths.

To examine the relationship between violation dynamics and market stress, we compute correlations between three time series: VaR submodularity violation rates, VaR subadditivity violation rates, and the CBOE Volatility Index (VIX). Table~\ref{tab:vix_corr} reports the Pearson correlation, the Spearman rank correlation, and the distance correlation of \citet{SRB07} for each window length and confidence level. The correlation between submodularity and subadditivity violation rates remains high across all configurations (Pearson 0.50--0.96, Spearman 0.46--0.90, distance 0.52--0.97), with the weakest link at the 250-day window and 99\% confidence (Pearson 0.50). This suggests that the two notions often fail during the same periods: when VaR violates submodularity, it also tends to violate subadditivity. Both violation rates also show positive correlation with the VIX. Pearson correlations between VIX and submodularity violation rates range from 0.22 to 0.50, while those between VIX and subadditivity violation rates range from 0.32 to 0.53, indicating that VaR diversification failures tend to coincide with periods of higher market volatility.

\begin{table}[htbp]
\centering
\caption{Correlations of VIX with VaR submodularity and subadditivity violation rates}
\label{tab:vix_corr}
\small
\begin{tabular}{cclccc}
\hline
Window & VaR Level & Correlation & VIX--Submod & VIX--Subadd & Submod--Subadd \\
\hline
250 & 95\% & Pearson  & 0.505 & 0.509 & 0.952 \\
    &        & Spearman & 0.230 & 0.195 & 0.854 \\
    &        & Distance & 0.511 & 0.507 & 0.941 \\[3pt]
250 & 97.5\% & Pearson  & 0.488 & 0.535 & 0.943 \\
    &        & Spearman & 0.406 & 0.314 & 0.845 \\
    &        & Distance & 0.521 & 0.529 & 0.941 \\[3pt]
250 & 99\% & Pearson  & 0.222 & 0.522 & 0.500 \\
    &        & Spearman & 0.148 & 0.243 & 0.463 \\
    &        & Distance & 0.283 & 0.512 & 0.521 \\[3pt]
500 & 95\% & Pearson  & 0.382 & 0.448 & 0.934 \\
    &        & Spearman & 0.078 & 0.159 & 0.885 \\
    &        & Distance & 0.389 & 0.430 & 0.902 \\[3pt]
500 & 97.5\% & Pearson  & 0.290 & 0.325 & 0.900 \\
    &        & Spearman & 0.411 & 0.375 & 0.901 \\
    &        & Distance & 0.373 & 0.337 & 0.896 \\[3pt]
500 & 99\% & Pearson  & 0.430 & 0.401 & 0.959 \\
    &        & Spearman & 0.321 & 0.233 & 0.702 \\
    &        & Distance & 0.426 & 0.418 & 0.971 \\
\hline
\multicolumn{6}{l}{\footnotesize $N=1{,}780$ trading days for all configurations.} \\
\end{tabular}
\end{table}

\section{Conclusion}
\label{sec:8}

Our main contribution is a systematic analysis of submodularity for law-invariant  risk measures, beyond the coherent setting studied by \citet{CC18}, with a new risk management interpretation and empirical illustrations.
Submodularity reflects a principle of rewards for ``flattening the risk",  and it turns out to give rise to rich and nontrivial mathematical results. We are able to fully characterize submodularity of   expected losses, certainty equivalents, law-invariant coherent risk measures, and several classical deviation measures.  
Within several families of convex risk measures with explicit formulas, we obtain a complete characterization for shortfall risk measures,  optimized certainty equivalents, adjusted Expected Shortfall, and monotone mean-deviation risk measures. The main message there is that the submodularity restriction behaves quite differently in each class.
Nevertheless, a full characterization of submodular convex risk measures, without assuming any specific form, is not yet available and requires further study.


\FloatBarrier
\bibliographystyle{apalike}
\bibliography{references}

@article{SRB07,
  author  = {Sz{\'e}kely, G{\'a}bor J. and Rizzo, Maria L. and Bakirov, Nail K.},
  title   = {Measuring and Testing Dependence by Correlation of Distances},
  journal = {The Annals of Statistics},
  year    = {2007},
  volume  = {35},
  number  = {6},
  pages   = {2769--2794}
}

@article{B22,
  author = {Bilmes, Jeff},
  title = {Submodularity in machine learning and artificial intelligence},
  journal = {arXiv preprint},
  pages ={arXiv:2202.00132},
  year = {2022}
}

@article{BMW22,
  author = {Burzoni, Matteo and Munari, Cosimo and Wang, Ruodu},
  title = {Adjusted {E}xpected {S}hortfall},
  journal = {Journal of Banking and Finance},
  year = {2022},
  volume = {134},
  pages = {106297}
}

@article{BT07,
  author = {Ben-Tal, Aharon and Teboulle, Marc},
  title = {An old-new concept of convex risk measures: The optimized certainty equivalent},
  journal = {Mathematical Finance},
  year = {2007},
  volume = {17},
  number = {3},
  pages = {449--476}
}

@article{BW22,
  author = {Bellini, Fabio and Fadina, Tolulope and Wang, Ruodu and Wei, Yunran},
  title = {Parametric Measures of Variability Induced by Risk Measures},
  journal = {Insurance: Mathematics and Economics},
  volume = {106},
  year = {2022},
  pages = {270--284},
  doi = {10.1016/j.insmatheco.2022.07.009}
}

@article{CC18,
  author = {Chateauneuf, Alain and Cornet, Bernard},
  title = {Choquet representability of submodular functions},
  journal = {Mathematical Programming Series B},
  year = {2018},
  volume = {168},
  pages = {615--629}
}

@article{CMMM11,
  title={Risk measures: rationality and diversification},
  author={Cerreia-Vioglio, Simone and Maccheroni, Fabio and Marinacci, Massimo and Montrucchio, Luigi},
  journal={Mathematical Finance},
  volume={21},
  number={4},
  pages={743--774},
  year={2011},
  publisher={Wiley Online Library}
}

@article{D89,
  title={Asset demand without the independence axiom},
  author={Dekel, Eddie},
  journal={Econometrica},
  volume={57},
  number={1},
  pages={163--169},
  year={1989},
}

@article{EMWW21,
  title={Bayes risk, elicitability, and the Expected Shortfall},
  author={Embrechts, Paul and Mao, Tiantian and Wang, Qiuqi and Wang, Ruodu},
  journal={Mathematical Finance},
  volume={31},
  number={4},
  pages={1190--1217},
  year={2021},
  publisher={Wiley Online Library}
}

@article{FS02,
  author = {F\"{o}llmer, Hans and Schied, Alexander},
  title = {Convex measures of risk and trading constraints},
  journal = {Finance and Stochastics},
  year = {2002},
  volume = {6},
  number = {4},
  pages = {429--447}
}

@book{FS16,
  author = {F\"ollmer, Hans and Schied, Alexander},
  title = {Stochastic Finance. An Introduction in Discrete Time},
  edition = {4th},
  publisher = {Walter de Gruyter},
  address = {Berlin},
  year = {2016}
}

@article{GG17,
  author = {Ghamami, Samim and Glasserman, Paul},
  title = {Submodular risk allocation},
  journal = {Management Science},
  year = {2019},
  volume = {65},
  number = {10},
  pages = {4656--4675}
}

@article{HWW26,
  author = {Han, Xia and Wang, Ruodu and Wu, Qiji},
  title = {Monotonic mean-deviation risk measures},
  journal = {Finance and Stochastics},
  year = {2026},
  volume = {30},
  pages = {441--483}
}

@book{MFE15,
  title={Quantitative Risk Management: Concepts, Techniques and Tools},
  author={McNeil, Alexander J and Frey, R{\"u}diger and Embrechts, Paul},
  year={2015},
  publisher={Princeton University Press},
  note ={Revised edition}
}

@incollection{MM04,
  author = {Marinacci, Massimo and Montrucchio, Luigi},
  title = {Introduction to the mathematics of ambiguity},
  booktitle = {Uncertainty in Economic Theory},
  editor = {Gilboa, Itzhak},
  publisher = {Routledge},
  address = {New York, NY, USA},
  year = {2004},
  pages = {46--107}
}

@article{MM08,
  author = {Marinacci, Massimo and Montrucchio, Luigi},
  title = {On concavity and supermodularity},
  journal = {Journal of Mathematical Analysis and Applications},
  year = {2008},
  volume = {344},
  number = {2},
  pages = {642--654}
}

@book{MS02,
  author = {M\"{u}ller, Alfred and Stoyan, Dietrich},
  title = {Comparison Methods for Stochastic Models and Risks},
  publisher = {Wiley},
  year = {2002}
}

@book{R13,
  author = {R{\"u}schendorf, Ludger},
  title = {Mathematical Risk Analysis. Dependence, Risk Bounds, Optimal Allocations and Portfolios},
  publisher = {Springer},
  address = {Heidelberg},
  year = {2013}
}

@article{S86,
  author = {Schmeidler, David},
  title = {Integral representation without additivity},
  journal = {Proceedings of the American Mathematical Society},
  year = {1986},
  volume = {97},
  number = {2},
  pages = {255--261}
}

@article{T78,
  author  = {Topkis, Donald M.},
  title   = {Minimizing a Submodular Function on a Lattice},
  journal = {Operations Research},
  year    = {1978},
  volume  = {26},
  number  = {2},
  pages   = {305--321},
  doi     = {10.1287/opre.26.2.305}
}

@article{TV20,
  author = {Tradacete, Pedro and Villanueva, Ignacio},
  title = {Valuations on {B}anach lattices},
  journal = {International Mathematics Research Notices},
  year = {2020},
  volume = {2020},
  number = {8},
  pages = {2468--2500}
}

@article{WWW20,
  author = {Wang, Ruodu and Wei, Yuchen and Willmot, Gordon E.},
  title = {Characterization, robustness and aggregation of signed {C}hoquet integrals},
  journal = {Mathematics of Operations Research},
  year = {2020},
  volume = {45},
  number = {3},
  pages = {993--1015}
}

@article{WWW20a,
  title={Distortion riskmetrics on general spaces},
  author={Wang, Qiuqi and Wang, Ruodu and Wei, Yunran},
  journal={ASTIN Bulletin},
  volume={50},
  number={3},
  pages={827--851},
  year={2020},
  publisher={Cambridge University Press}
}

@article{RUZ06,
  title={Generalized deviations in risk analysis},
  author={Rockafellar, R Tyrrell and Uryasev, Stan and Zabarankin, Michael},
  journal={Finance and Stochastics},
  volume={10},
  number={1},
  pages={51--74},
  year={2006},
  publisher={Springer}
}

@article{Y87,
  author = {Yaari, Menahem E.},
  title = {The dual theory of choice under risk},
  journal = {Econometrica},
  year = {1987},
  volume = {55},
  number = {1},
  pages = {95--115}
}

\appendix
\renewcommand{\thesubsection}{\thesection\arabic{subsection}}
\section{Discussion on submodularity for sets of assets}
\label{sec:6}

We briefly discuss another form of submodularity, formulated on a different lattice, and its relation to our setting.
As noted by \citet{GG17} and \citet{B22}, submodularity is the property of diminishing marginal risk: the marginal change in risk from adding an asset to a portfolio decreases with the addition of another asset. Mathematically, one considers functions $v$ that take subsets of $[n]:=\{1,\dots,n\}$ as the input variable, which represent the set of assets included in a portfolio. 
Such a function $v$
is submodular if it satisfies
\begin{equation}
\label{eq:v-submo1}
v(S\cup T) + v(S\cap T) \le v(T) + v(S), \qquad S, T \subseteq [n].
\end{equation}
This is submodularity of $v$ formulated on the lattice $ (2^{[n]}, \cup,\cap)$. 
It is well known that \eqref{eq:v-submo1} is
equivalent to 
\begin{equation}
\label{eq:v-submo}
v(S\cup{\{z\}}) - v(S) \ge v(T\cup{\{z\}}) - v(T), \qquad S\subseteq T \subseteq [n],  \qquad \forall z \notin T.
\end{equation}
For a given risk measure $\rho$ and a random vector $(X_i)_{i\in [n]}$,
write $v(S) = \rho(\sum_{i \in S} X_i)$ for $S\subseteq [n]$.
The submodularity of $v$ as defined in \eqref{eq:v-submo} means that,  for $S \subseteq T \subseteq [n]$ and any $j \notin T$,  we have
$$
\;\rho\!\left(\sum_{i\in S} X_i + X_j\right)
\;-\;
\rho\!\left(\sum_{i\in S} X_i\right)
\;\ge\;
\rho\!\left(\sum_{i\in T} X_i + X_j\right)
\;-\;
\rho\!\left(\sum_{i\in T} X_i\right).\;
$$
Rearranging terms, it is
\begin{equation}
    \label{eq:rho-submo}
\rho \left(\sum_{i\in T} X_i + X_j\right)  +
\rho \left(\sum_{i\in S} X_i\right)
\le 
 \rho \left(\sum_{i\in S} X_i + X_j\right)
 +
\rho \left(\sum_{i\in T} X_i\right).
\end{equation} 
The property in \eqref{eq:rho-submo}
is different from our definition of submodularity, but there is a connection. 
Suppose that $(X_i)_{i\in [n]}$
has nonnegative components and
$\p(X_j>0,\sum_{i\in T}X_i>0)=0$.
Write $X=\sum_{i\in S}X_i+X_j$
and  $Y=\sum_{i\in T}X_i$,
then we have $X\vee Y=\sum_{i\in T}X_i +X_j $
and $X\wedge Y=\sum_{i\in S}X_i$. Therefore, \eqref{eq:rho-submo} becomes
$$
\rho(X\vee Y) + \rho(X\wedge Y) \le \rho(X) + \rho(Y),
$$
thus our definition of submodularity for this particular choice of $(X,Y)$. In general, the two notions of submodularity do not imply each other, and they carry different economic meanings.

\section{Some omitted proofs}
\label{app:proofs}

We present the proofs of the two omitted cases in Theorems \ref{thm:linear-iff} and \ref{thm:AES-submodular-characterization}.

\subsection{Theorem \ref{thm:linear-iff} without differentiability}
\label{app:SR-proof}

We now extend Theorem~\ref{thm:linear-iff} to loss functions $\ell$ that are not necessarily twice differentiable. Denote by $S=\ell_+'$ the right derivative of $\ell$,
$$
S(u)=\lim_{h\downarrow 0}\frac{\ell(u+h)-\ell(u)}{h},
$$
which exists and is strictly positive for every $u$ by convexity of $\ell$.

Since $S$ is positive and increasing, $\log S$ is well-defined and increasing. To handle points where $\log S$ is not differentiable, we work with one-sided Dini derivatives. For an extended-real function $f:\R\to[-\infty,\infty]$, define
\begin{align*}
D^+ f(x)&=\limsup_{h\downarrow 0}\frac{f(x+h)-f(x)}{h},\\
D_-^{\mathrm{up}} f(x)&=\limsup_{h\downarrow 0}\frac{f(x)-f(x-h)}{h},\\
D_-^{\mathrm{low}} f(x)&=\liminf_{h\downarrow 0}\frac{f(x)-f(x-h)}{h},
\end{align*}
and set
$$
R_+(x)=D^+(\log S)(x),\qquad R_-^{\mathrm{up}}(x)=D_-^{\mathrm{up}}(\log S)(x),\qquad R_-^{\mathrm{low}}(x)=D_-^{\mathrm{low}}(\log S)(x).
$$
Define
$$
R(x)=\min\{R_+(x),\,R_-^{\mathrm{low}}(x)\}\in[0,\infty].
$$
The choice of $R_-^{\mathrm{low}}$ (rather than $R_-^{\mathrm{up}}$) is dictated by the lower bound in Lemma~\ref{lem:dini-weight} below. Since $\log S$ is increasing, $R(x)\ge 0$ for all $x$; moreover, $R(x)=\infty$ at any jump of $\log S$. Set
$$
L=\inf_{x\in\R}R(x),\qquad h(x)=S(x)\big(R(x)-2L\big).
$$
\begin{theorem}\label{thm:linear-iff-app}
For a strictly increasing loss function $\ell$, the shortfall risk measure $\rho_\ell$ is submodular if and only if there exists $\lambda\in\R$ such that
\begin{equation}\label{eq:LD-app}
h(x)\le \lambda\,\ell(x),\qquad x\in\R.
\end{equation}
\end{theorem}

\begin{proof}
We follow the same route as in the differentiable case: prove sufficiency on each simple subspace and then invoke Lemma \ref{lem:bounded-extension}.

\textbf{Sufficiency.}
Assume \eqref{eq:LD-app} holds. It is enough to show that $\rho_\ell$ is submodular on every simple subspace.

Let $\mathcal F'=\sigma(A_1,\dots,A_n)$ be a simple sub-$\sigma$-algebra, where $A_1,\dots,A_n$ are its atoms and
$$
\pi_k=\mathbb P(A_k)>0,\qquad \sum_{k=1}^n \pi_k=1.
$$
For
$$
X=\sum_{k=1}^n x_k\mathbf 1_{A_k},
$$
identify $X$ with $\bx=(x_1,\dots,x_n)\in\R^n$ and write $m(\bx)=\rho_\ell(X)$. Since $\rho_\ell$ is monotone and cash invariant, $m$ is $1$-Lipschitz in the sup norm. The defining equation $\E[\ell(X-\rho_\ell(X))]=0$ and $\ell(0)=0$ give
\begin{equation}\label{eq:m-satisfy-weighted}
\sum_{k=1}^{n} \pi_k\,\ell(x_k-m(\bx))=0,\qquad \bx\in \R^n.
\end{equation}
Since $\ell$ is finite and $S>0$, the inequality forces $R(x)<\infty$ for every $x$, so $\log S$ has no jumps and $S$ is continuous.

Differentiating \eqref{eq:m-satisfy-weighted} from the right in direction $e_j$ and solving for $\nabla^+_j m$ yields the right-gradient formula
\begin{equation}\label{eq:gradient}
\nabla^+_j m(\bx)=w_j(\bx),\qquad j\in[n],
\end{equation}
where
$$
T(\bx)=\sum_{k=1}^{n} \pi_k S(x_k-m(\bx)),\qquad
w_j(\bx)=\frac{\pi_j S(x_j-m(\bx))}{T(\bx)}.
$$

Fix $i\ne j$ and write $y_k=x_k-m(\bx)$. Since $S$ is monotone, it is differentiable almost everywhere, with $S'(y_k)=S(y_k)R(y_k)$ at differentiability points. Differentiating $w_i$ in direction $e_j$ via \eqref{eq:gradient} gives, for almost every $\bx$,
\begin{equation}\label{eq:phi-derivative}
\nabla^+_j\nabla^+_i m(\bx)
=-w_iw_j\Big(R(y_i)+R(y_j)-\sum_{k=1}^n w_k R(y_k)\Big).
\end{equation}
Evaluating \eqref{eq:LD-app} at the points $y_k$ and summing with weights $\pi_k$ yields
$$
\sum_{k=1}^n \pi_k h(y_k)\le \lambda\sum_{k=1}^n \pi_k \ell(y_k)=0,
$$
that is,
$$
\sum_{k=1}^n w_k R(y_k)\le 2L.
$$
Since $R(y_i)+R(y_j)\ge 2L$, the parenthesized expression in \eqref{eq:phi-derivative} is nonnegative. Therefore,
$$
\nabla_j^+\nabla_i^+m(\bx)\le 0 \qquad \text{for almost every }\bx.
$$
By Theorem~3.9.3(a) in \citet{MS02}, this implies that $m$ is submodular on $\R^n$, and hence $\rho_\ell$ is submodular on $L^\infty(\Omega,\mathcal F',\mathbb P)$.

Since $\mathcal F'$ was arbitrary, $\rho_\ell$ is submodular on every simple subspace. Lemma \ref{lem:bounded-extension} now yields submodularity on all of $L^\infty$.

The necessity proof relies on the following lemma, which provides a one-sided analogue of \eqref{eq:phi-derivative} without differentiability.

As in the differentiable case, for the necessity direction we work on a uniform $n$-atom space $([n],2^{[n]},\mathbb P)$ with equal atom weight $1/n$. We identify each random variable $X$ with the vector $\bx=(X(1),\dots,X(n))\in\R^n$ and write $m(\bx)=\rho_\ell(X)$.
 In this setting, the defining equation becomes
\begin{equation}\label{eq:m-satisfy}
\sum_{k=1}^{n} \ell(x_k-m(\bx))=0,\qquad \bx\in \R^n.
\end{equation}

\begin{lemma}\label{lem:dini-weight}
Assume $S:\R\to(0,\infty)$ is continuous and increasing. Let $\bz\in\R^n$ satisfy $m(\bz)=0$, fix $i\neq j$, and set
$$
K=\sum_{k=1}^n S(z_k),\qquad w_k=w_k(\bz)=\frac{S(z_k)}{K}.
$$
Assume $R_+(z_i)$ is finite, and that $R_-^{\mathrm{up}}(z_k)$ and $R_-^{\mathrm{low}}(z_k)$ are finite for every $k\neq i$. For $t\ge 0$, write $\bz(t)=\bz+t e_i$ and $w_k(t)=w_k(\bz(t))$. Then
\begin{equation}\label{eq:dini-weight}
\liminf_{t\downarrow0}\frac{w_j(t)-w_j}{t}
\ge -w_iw_j\Big(
R_+(z_i)+R_-^{\mathrm{up}}(z_j)-w_iR_+(z_i)
-\sum_{k\neq i} w_k\,R_-^{\mathrm{low}}(z_k)\Big).
\end{equation}
\end{lemma}

\begin{proof}
Since $m$ is $1$-Lipschitz and $m(\bz)=0$, we have $0\le m(t)=m(\bz(t))\le t$ along the path $\bz(t)$, and the right-gradient formula gives $m(t)/t\to w_i$ as $t\downarrow 0$.

Set $\delta_k(t)=\log S(z_k(t)-m(t))-\log S(z_k)$. Then
$$
\log w_j(t)-\log w_j
=\delta_j(t)-\log\!\Big(\sum_{k=1}^n w_k\,e^{\delta_k(t)}\Big).
$$
Using $e^x-1\ge x$ in $w_j(t)=w_j\exp(\log w_j(t)-\log w_j)$ gives
\begin{equation}\label{eq:stepA}
\frac{w_j(t)-w_j}{t}\ge w_j\bigg(
\frac{\delta_j(t)}{t}-\frac{1}{t}\log\!\Big(\sum_{k=1}^n w_k\,e^{\delta_k(t)}\Big)\bigg).
\end{equation}

For the first term: since $z_j(t)=z_j$ and $m(t)/t\to w_i$,
$$
\liminf_{t\downarrow0}\frac{\delta_j(t)}{t}
=\liminf_{t\downarrow0}\frac{\log S(z_j-m(t))-\log S(z_j)}{t}\ge -w_i R_-^{\mathrm{up}}(z_j).
$$

For the second term: the finiteness assumptions on the Dini derivatives imply the required local first-order control near $t=0$. Hence there exist constants $C>0$ and $t_0>0$ such that
$$
\log\!\Big(\sum_{k=1}^n w_k e^{\delta_k(t)}\Big)\le \sum_{k=1}^n w_k\delta_k(t)+Ct^2,\qquad 0<t<t_0,
$$
which yields
$$
\limsup_{t\downarrow0}\frac{1}{t}\log\!\Big(\sum_{k=1}^n w_k\,e^{\delta_k(t)}\Big)
\le \sum_{k=1}^n w_k\limsup_{t\downarrow0}\frac{\delta_k(t)}{t}.
$$
For $k=i$, since $z_i(t)=z_i+t$ and $m(t)/t\to w_i$, we have $\limsup_{t\downarrow 0}\delta_i(t)/t\le (1-w_i)R_+(z_i)$. For $k\neq i$, since $z_k(t)=z_k$, we have $\limsup_{t\downarrow 0}\delta_k(t)/t\le -w_i\,R_-^{\mathrm{low}}(z_k)$. Therefore
$$
\limsup_{t\downarrow0}\frac{1}{t}\log\!\Big(\sum_{k=1}^n w_k\,e^{\delta_k(t)}\Big)
\le w_i(1-w_i)R_+(z_i)-w_i\sum_{k\neq i}w_k\,R_-^{\mathrm{low}}(z_k).
$$

Substituting both estimates into \eqref{eq:stepA} and rearranging gives \eqref{eq:dini-weight}.
\end{proof}

\textbf{Necessity.}
Assume $\rho_\ell$ is submodular, so $m$ is submodular on $\R^n$ for every $n\ge 3$. The proof proceeds in three steps.

\emph{Step~1: $S$ is continuous.}
Suppose for contradiction that $S$ has a jump at some point, which by cash invariance we may take to be $0$:
$$
s_-=S(0^-)< S(0^+)=s_+.
$$
Consider $\bx=(-2h,-h,0)$ and $\mathbf{y}=(-h,-h,-h)=-h\mathbf{1}$ in $\R^3$, so that
$$
\bx\vee\mathbf{y}=(-h,-h,0),\qquad \bx\wedge\mathbf{y}=(-2h,-h,-h).
$$
Write $m(\bx)=-\alpha_h h$, $m(\bx\vee\mathbf{y})=-\beta_h h$, $m(\bx\wedge\mathbf{y})=-\gamma_h h$, and note $m(\mathbf{y})=-h$. Dividing the defining equation \eqref{eq:m-satisfy} by $h$ and sending $h\downarrow 0$, using the left and right derivatives $s_-$ and $s_+$ of $\ell$ at $0$, yields limiting values
$$
\alpha=\frac{3s_-}{2s_-+s_+},\qquad
\beta=\frac{2s_-}{2s_-+s_+},\qquad
\gamma=\frac{2(s_-+s_+)}{s_-+2s_+}.
$$
The submodularity deficit is
$$
\Delta_h= m(\bx)+m(\mathbf{y})-m(\bx\vee\mathbf{y})-m(\bx\wedge\mathbf{y})
= (\beta_h+\gamma_h-\alpha_h-1)\,h,
$$
and passing to the limit gives
$$
\frac{\Delta_h}{h}\to
\beta+\gamma-\alpha-1
=\frac{s_-(s_--s_+)}{(s_-+2s_+)(2s_-+s_+)}<0,
$$
contradicting submodularity. Hence $S$ is continuous on $\R$.

Since $S$ is continuous, positive, and increasing, $\log S$ is continuous and monotone. By Lebesgue's differentiation theorem, $R(x)<\infty$ for almost every $x$.

\emph{Step~2: A balanced-sum inequality.}
We show that
\begin{equation}\label{eq:balanced-app}
\sum_{k=1}^r h(x_k)\le 0 \qquad\text{whenever}\quad \sum_{k=1}^r \ell(x_k)=0.
\end{equation}

Fix $\varepsilon>0$ and choose $v$ such that $R(v)\le L+\varepsilon$ and $S$ is differentiable at $v$; this is possible since the set of differentiability points has full measure and $L=\inf R$.
Pick $p\in\R$ with $\ell(p)\neq 0$ and $\ell(p)\ell(v)<0$, and let
$M=\lfloor -2\ell(v)/\ell(p)\rfloor$,
so that the residual $d=-2\ell(v)-M\ell(p)$ satisfies $d\,\ell(p)\ge 0$ and $|d|<|\ell(p)|$. By continuity and strict monotonicity of $\ell$, there is a unique $c$ with $\ell(c)=d$.

Now fix $x_1,\dots,x_r$ with $\sum_{k=1}^r \ell(x_k)=0$. For each $N\in\mathbb N$, form the vector $\bz^{(N)}\in\R^n$ (with $n=M+Nr+3$) consisting of two copies of $v$, $M$ copies of $p$, one copy of $c$, and $N$ copies of $(x_1,\dots,x_r)$. By construction, $\sum_k \ell(z_k^{(N)})=0$, hence $m(\bz^{(N)})=0$.

Apply submodularity to the pair of $v$-entries: for $t,s>0$,
$$
m(\bz^{(N)})+m(\bz^{(N)}+te_i+se_j)\le m(\bz^{(N)}+te_i)+m(\bz^{(N)}+se_j).
$$
Dividing by $s$, sending $s\downarrow 0$ (using the right-gradient formula), then dividing by $t$ and sending $t\downarrow 0$ yields
$$
\liminf_{t\downarrow0}\frac{w_j(\bz^{(N)}+te_i)-w_j(\bz^{(N)})}{t}\le 0.
$$
Lemma~\ref{lem:dini-weight} provides a matching lower bound. Since $z_i^{(N)}=z_j^{(N)}=v$ and $S$ is differentiable at $v$, all one-sided Dini derivatives at $v$ coincide with $R(v)$. (For coordinates $k\neq i$, the finiteness assumptions of the lemma hold after an arbitrarily small perturbation, which we then let vanish.) Combining the upper and lower bounds gives
$$
\sum_{k=1}^n w_k(\bz^{(N)}) R(z_k^{(N)}) \le 2R(v).
$$
Multiplying by $K=\sum_k S(z_k^{(N)})$ and isolating the $N$ repeated blocks,
$$
N\sum_{k=1}^r S(x_k)\big(R(x_k)-2R(v)\big)\le 2S(v)R(v) - A,
$$
where $A=MS(p)(R(p)-2R(v))+S(c)(R(c)-2R(v))$ is independent of $N$. Dividing by $N$ and letting $N\to\infty$ yields $\sum_{k=1}^r S(x_k)(R(x_k)-2R(v))\le 0$. Since $R(v)\le L+\varepsilon$,
$$
\sum_{k=1}^r h(x_k)
=\sum_{k=1}^r S(x_k)\big(R(x_k)-2L\big)
\le 2\varepsilon\sum_{k=1}^r S(x_k),
$$
and sending $\varepsilon\downarrow 0$ gives \eqref{eq:balanced-app}.

\emph{Step~3: From the balanced-sum inequality to \eqref{eq:LD-app}.}
Define $\phi(x)=h(x)/\ell(x)$ for $\ell(x)\neq 0$ and set
$$
\alpha^+ = \sup_{\ell(x)>0}\phi(x),
\qquad
\alpha^- = \inf_{\ell(x)<0}\phi(x).
$$
We claim $\alpha^+\le \alpha^-$, both finite; any $\lambda\in[\alpha^+,\alpha^-]$ then satisfies \eqref{eq:LD-app} for all $x$ with $\ell(x)\neq 0$, and \eqref{eq:balanced-app} applied to the singleton $\{0\}$ gives $h(0)\le 0$, which is \eqref{eq:LD-app} at $x=0$.

Fix $a,b\in\R$ with $\ell(a)<0<\ell(b)$ and set $\theta=\ell(b)/(\ell(b)-\ell(a))\in(0,1)$, so that $\theta\ell(a)+(1-\theta)\ell(b)=0$.

By Dirichlet's approximation theorem (or directly, if $\theta$ is rational), there exist infinitely many pairs $(r_N,N)$ with $1\le r_N\le N-1$ and $|\theta-r_N/N|<1/N^2$. Set $s_N=N-r_N$ and
$$
\delta_N = -(r_N\ell(a)+s_N\ell(b))=(\ell(b)-\ell(a))(r_N-\theta N),
$$
so $|\delta_N|\le |\ell(b)-\ell(a)|/N\to 0$.

Fix $A>0$ with $[-2A,2A]\subset\ell((-1,1))$ and set $I=\ell^{-1}([-2A,2A])\subset(-1,1)$, $I_0=\ell^{-1}([-A,A])\subset I$. For $N$ large enough that $|\delta_N|\le A$, define
$$
T_N(x)=\ell^{-1}(\delta_N-\ell(x)),\qquad x\in I_0,
$$
so that $\ell(x)+\ell(T_N(x))=\delta_N$ and $T_N(I_0)\subset I$. Let $M_N=\sqrt{N}$ and
$$
H_N=\{x\in I:h(x)\text{ is finite and }|h(x)|\le M_N\}.
$$
Since $h$ is finite almost everywhere and $M_N\to\infty$, we have $|I\setminus H_N|\to 0$.

Since $S$ is bounded between positive constants $m\le M$ on the compact interval $I$, the maps $\ell$ and $\ell^{-1}$ are bi-Lipschitz on $I$ and $\ell(I)$, hence $T_N^{-1}$ is $K$-Lipschitz with $K=M/m$. Therefore $|T_N^{-1}(I\setminus H_N)|\le K|I\setminus H_N|$, and
$$
|I_0\setminus(H_N\cap T_N^{-1}(H_N))|
\le (1+K)|I\setminus H_N|\to 0.
$$
For large $N$, the set $H_N\cap T_N^{-1}(H_N)\cap I_0$ has positive measure. Pick $c_N$ in this set and set $d_N=T_N(c_N)$. Then $c_N,d_N\in H_N$, so $|h(c_N)|,|h(d_N)|\le\sqrt{N}$, and
$$
r_N\ell(a)+s_N\ell(b)+\ell(c_N)+\ell(d_N)=0.
$$
Applying \eqref{eq:balanced-app} to this multiset ($r_N$ copies of $a$, $s_N$ copies of $b$, and the pair $c_N,d_N$), dividing by $N$, and letting $N\to\infty$ (the correction $(h(c_N)+h(d_N))/N\to 0$) gives
$$
\theta\, h(a)+(1-\theta)\,h(b)\le 0,
$$
which rearranges to $\phi(b)\le \phi(a)$, that is, $h(b)/\ell(b)\le h(a)/\ell(a)$. Taking the supremum over $b$ and infimum over $a$ yields $\alpha^+\le\alpha^-$. Finiteness follows by fixing one of $a$ or $b$. This completes the proof.
\end{proof}

\subsection[Theorem \ref{thm:AES-submodular-characterization} with g(1-)<infty]{Theorem \ref{thm:AES-submodular-characterization} with $g(1-)<\infty$}
\label{app:AES-proof}

It remains to prove the ``if'' statement under the assumption $g(1-)<\infty$. 

\begin{proof}[Proof of Theorem \ref{thm:AES-submodular-characterization} with $g(1-)<\infty$]
Write 
$\rho=\ES^g$.   For Cases 1 and 2 below, we need to  construct $X$ and $Y$ as follows.
Let $a>0$, $q\in[0,1]$, and $b\in\R$ be three numbers, which will be
determined later. Let $U$ be uniformly distributed on $[0,1]$, and let
$V=U\id_{\{U\ge q\}}+(q-U)\id_{\{U<q\}}$, which is also uniformly distributed
on $[0,1]$. Define the random variables
$$
X=2aU+b-a,\qquad Y=2aV+b-a.
$$
It is easy to compute that $\ES_p(X)=\ES_p(Y)=ap+b$ for $p\in[0,1]$.
Moreover, for $p\in[q,1]$, we have
$$
\VaR_p(X\vee Y)=\VaR_p(X)=2ap+b-a,
$$
and for $p\in[0,q)$, we have
$$
\VaR_p(X\vee Y)=\VaR_{(p+q)/2}(X)=a(p+q)+b-a.
$$
Hence, for $p\in[q,1]$, we have $\ES_p(X\vee Y)=ap+b$, and for $p\in[0,q)$,
\begin{align*}
\ES_p(X\vee Y)
&= \frac{a}{1-p}\left(\int_p^q (r+q)\,\d r + \int_q^1 2r\,\d r \right) + (b-a)\\
&= ap+b+\frac{a(q-p)^2}{2(1-p)}.
\end{align*}
Hence, for $p\in[0,q)$,
 $$
\ES_p(X\vee Y)-\ES_p(X)
 = \frac{a (q-p)^2 }{2 (1-p)} >0.  $$
Also, $X\wedge Y=X=Y$ on $\{U\ge q\}$, while
$$
X\wedge Y=2a\min\{U,q-U\}+b-a\le aq+b-a
$$
on $\{U<q\}$. Since $X\wedge Y=2aU+b-a\ge aq+b-a$ on $\{U\ge q\}$, the
largest $1-q$ fraction of values of $X\wedge Y$ is carried by $\{U\ge q\}$,
and therefore
$$
\ES_q(X\wedge Y)=\ES_q(X)=aq+b.
$$

\paragraph{Case 1.} Suppose that $g$ is not convex on $[0,1)$.
Since $p\mapsto \ES_p(X)$ is continuous for every $X\in L^\infty$, we may
assume without loss of generality that $g$ is lower semicontinuous.

We now choose the parameters $a,b,q$ in the construction above and show
that $\rho$ is not submodular. Let
$$
p^*=\sup\{p\in [0,1]: g(p)<\infty\}
$$
and let $g^*:[0,1]\to (-\infty,\infty]$ be the largest
$(-\infty,\infty]$-valued convex function on $[0,1]$ dominated by $g$.
Clearly, $g^*(p)=\infty$ for $p\in (p^*,1]$.

Suppose that $g$ is not convex on $[0,p^*]$. There exist distinct points
$p_1,p_2\in [0,p^*]$ such that $g^*(p_1)=g(p_1)$, $g^*(p_2)=g(p_2)$, and
$g^*$ is linear and strictly increasing on $[p_1,p_2]$. Choose $q=p_2$
and let $a,b$ be such that the
affine function $p\mapsto ap+b$ coincides with $g^*$ on
$[p_1,p_2]$, and clearly $g^*\ge ap+b$ due to convexity. By using
$$
aq+b=\ES_q(X)=\ES_q(Y)=\ES_q(X\wedge Y)=g(q),
$$
we have
\begin{align*}
\rho(X) &= \rho(Y)=\sup_{p\in [0,1]} \left\{ ap+b-g(p) \right\}= 0,\\
\rho(X\wedge Y)&\ge \ES_q(X\wedge Y)-g(q)=0,\\
\rho(X \vee Y)&\ge \ES_{p_1}(X\vee Y)-g(p_1)=\frac{a (q-p_1)^2 }{2 (1-p_1)} >0.
\end{align*}
Therefore,
$$
\rho(X \vee Y)+ \rho(X \wedge Y) > 0= \rho(X) + \rho(Y),
$$
showing that $\rho$ is not submodular, a contradiction.

\paragraph{Case 2.} Assume that $g$ is convex on $[0,1)$. We show that $g$ is constant on $[0,1)$. Suppose for
contradiction that $g$ is not constant on $[0,1)$. Since $g$ is
increasing, there exists $q_0<1$ such that $g(q_0)>0$.

Let $m=g'_+$ be the right derivative of $g$ on $(0,1)$. Since $g$ is
convex on $[0,1)$, the function $m$ is increasing on $(0,1)$. Moreover,
$$
g(x)=\int_0^x m(t)\,\d t,\qquad x\in[0,1),
$$
and therefore $m$ is integrable on $(0,1)$ because $g(1-)<\infty$.
Since $g(q_0)>0$,
$$
0<g(q_0)=\int_0^{q_0} m(t)\,\d t.
$$
Hence there exists $x_0\in(0,q_0)$ such that $m(x_0)>0$. Since $m$ is
increasing, we have $m(x)\ge m(x_0)>0$ for every $x\in[x_0,1)$. Moreover,
since $m$ is increasing, it is differentiable almost everywhere on $(x_0,1)$.

We claim that there exists a point $q\in(x_0,1)$ at which $m$ is
differentiable and
$$
m'(q)<\dfrac{m(q)}{1-q}.
$$

Indeed, suppose for contradiction that
$$
m'(x)\ge \frac{m(x)}{1-x}
$$
for almost every $x\in(x_0,1)$. Since $m$ is increasing, its distributional
derivative $\d m$ is a nonnegative measure on $(x_0,1)$. The above
inequality implies
$$
\d m\ge \frac{m(x)}{1-x}\,\d x
$$
as measures on $(x_0,1)$. Now define
$$
H(x)=(1-x)m(x),\qquad x\in(x_0,1).
$$
By the product rule for distributional derivatives,
$$
\d H=(1-x)\,\d m-m(x)\,\d x\ge 0.
$$
Hence $H$ is increasing on $(x_0,1)$. Since $H(x_0)>0$, it follows that
$$
(1-x)m(x)\ge H(x_0)>0,\qquad x\in(x_0,1),
$$
that is,
$$
m(x)\ge \frac{H(x_0)}{1-x},\qquad x\in(x_0,1).
$$
Therefore,
$$
\int_{x_0}^1 m(x)\,\d x
\ge
\int_{x_0}^1 \frac{H(x_0)}{1-x}\,\d x
=\infty,
$$
contradicting the fact that $m$ is integrable on $(0,1)$. The claim follows.

Fix such a point $q$, and set $a=m(q)>0$ and $b=g(q)-aq$. Since $m$ is
differentiable at $q$, it is continuous at $q$, and thus $g$ is
differentiable at $q$ with $g'(q)=a$. Hence $\ell(p)=ap+b$ for
$p\in[0,1]$ is the tangent line to $g$ at $q$, so $\ell(p)\le g(p)$ for
$p\in[0,1)$. Letting $p\uparrow1$, we obtain $\ell(1)\le g(1-)\le g(1)$,
and hence $\ell(p)\le g(p)$ for all $p\in[0,1]$.
Applying the construction above with this choice of $a,b,q$, we obtain
$$
\rho(X)=\rho(Y)=\sup_{p\in[0,1]}\{\ell(p)-g(p)\}=0.
$$
Thus
$$
\rho(X\wedge Y)\ge \ES_q(X\wedge Y)-g(q)=0.
$$

Now let $p=q-h$ with $h\downarrow0$. Since $m$ is differentiable at $q$,
$m(q-h)=a-m'(q)h+o(h)$. Using
$$
g(q)-g(q-h)=\int_{q-h}^q m(t)\,\d t,
$$
we obtain
$$
g(q-h)=g(q)-ah+\frac{m'(q)}{2}h^2+o(h^2).
$$
Since $g(q)=aq+b$ and $\ell(p)=ap+b$,
$$
g(q-h)-\ell(q-h)=\frac{m'(q)}{2}h^2+o(h^2),
$$
and therefore
$$
\begin{aligned}
\ES_{q-h}(X\vee Y)-g(q-h)
&=\ES_{q-h}(X)-g(q-h)+\frac{ah^2}{2(1-q+h)}\\
&=\ell(q-h)-g(q-h)+\frac{ah^2}{2(1-q+h)}\\
&=\frac12\left(\frac{a}{1-q+h}-m'(q)\right)h^2+o(h^2).
\end{aligned}
$$
Since $m'(q)<a/(1-q)$, the right-hand side is strictly positive for all
sufficiently small $h>0$. Thus $\rho(X\vee Y)>0$, and hence
$$
\rho(X\vee Y)+\rho(X\wedge Y) > 0 = \rho(X)+\rho(Y),
$$
contradicting submodularity. Therefore $g$ is constant on $[0,1)$, and for
every $X\in L^\infty$,
$$
\rho(X)=\sup_{p\in[0,1)}\ES_p(X)=\ES_1(X).
$$
Thus the conclusion holds with $p_0=1$. 
\end{proof}

\FloatBarrier
\section{Stock selection}\label{app:stocks}
In the table below we present the list of  stocks used in Section \ref{sec:7}.
\begin{table}[H]
\centering
\caption{Historical S\&P~500 stock selection by sector and market capitalization}
\label{tab:stock_selection}
{\scriptsize
\setlength{\tabcolsep}{3.5pt}
\renewcommand{\arraystretch}{0.92}
\begin{tabular}{@{}lllr@{}}
\toprule
GICS Sector & Ticker & Company & Mkt Cap (\$B) \\
\midrule
\multirow[t]{5}{*}{Consumer Discretionary} & AMZN  & Amazon & 632.1 \\
                                           & DIS   & Walt Disney & 184.8 \\
                                           & HD    & Home Depot & 152.5 \\
                                           & CMCSA & Comcast & 108.5 \\
                                           & MCD   & McDonald's & 101.3 \\
\addlinespace[1pt]
\multirow[t]{5}{*}{Consumer Staples}       & WMT   & Walmart & 230.9 \\
                                           & PG    & Procter \& Gamble & 172.1 \\
                                           & KO    & Coca-Cola & 153.5 \\
                                           & PEP   & PepsiCo & 126.3 \\
                                           & PM    & Philip Morris & 106.1 \\
\addlinespace[1pt]
\multirow[t]{5}{*}{Energy}                 & XOM   & ExxonMobil & 237.9 \\
                                           & CVX   & Chevron & 169.5 \\
                                           & SLB   & Schlumberger & 77.4 \\
                                           & OXY   & Occidental Petroleum & 57.8 \\
                                           & COP   & ConocoPhillips & 50.8 \\
\addlinespace[1pt]
\multirow[t]{5}{*}{Financials}             & BRK-B & Berkshire Hathaway & 426.6 \\
                                           & JPM   & JPMorgan Chase & 229.0 \\
                                           & BAC   & Bank of America & 177.4 \\
                                           & WFC   & Wells Fargo & 150.4 \\
                                           & C     & Citigroup & 99.4 \\
\addlinespace[1pt]
\multirow[t]{5}{*}{Health Care}            & JNJ   & Johnson \& Johnson & 270.6 \\
                                           & UNH   & UnitedHealth & 176.2 \\
                                           & PFE   & Pfizer & 135.6 \\
                                           & ABBV  & AbbVie & 121.4 \\
                                           & MRK   & Merck & 102.6 \\
\addlinespace[1pt]
\multirow[t]{5}{*}{Industrials}            & BA    & Boeing & 223.0 \\
                                           & GE    & General Electric & 82.4 \\
                                           & MMM   & 3M & 78.5 \\
                                           & HON   & Honeywell & 74.9 \\
                                           & UPS   & United Parcel Service & 73.0 \\
\addlinespace[1pt]
\multirow[t]{5}{*}{Information Technology} & GOOGL & Alphabet (Class A) & 644.0 \\
                                           & AAPL  & Apple & 599.6 \\
                                           & MSFT  & Microsoft & 583.3 \\
                                           & V     & Visa & 204.1 \\
                                           & INTC  & Intel & 199.4 \\
\addlinespace[1pt]
\multirow[t]{5}{*}{Materials}              & ECL   & Ecolab & 34.8 \\
                                           & SHW   & Sherwin-Williams & 31.8 \\
                                           & NEM   & Newmont & 31.6 \\
                                           & IFF   & International Flavors \& Fragrances & 31.6 \\
                                           & APD   & Air Products and Chemicals & 29.7 \\
\addlinespace[1pt]
\multirow[t]{5}{*}{Real Estate}            & AMT   & American Tower & 53.2 \\
                                           & PLD   & Prologis & 46.6 \\
                                           & EQIX  & Equinix & 37.5 \\
                                           & SPG   & Simon Property Group & 34.2 \\
                                           & O     & Realty Income & 34.0 \\
\addlinespace[1pt]
\multirow[t]{2}{*}{Telecommunication Services} & VZ    & Verizon Communications & 142.3 \\
                                               & T     & AT\&T & 114.1 \\
\addlinespace[1pt]
\multirow[t]{5}{*}{Utilities}              & PCG   & PG\&E & 111.1 \\
                                           & NEE   & NextEra Energy & 65.1 \\
                                           & D     & Dominion Energy & 49.1 \\
                                           & DUK   & Duke Energy & 47.0 \\
                                           & SO    & Southern Company & 39.9 \\
\bottomrule
\end{tabular}
}
\begin{minipage}{0.95\textwidth}
\footnotesize Approximate market capitalizations on January 2, 2018 are computed from historical shares outstanding and adjusted closing prices via the \texttt{yfinance} API; the constituent snapshot is compiled from publicly available historical records of index membership.
\end{minipage}
\end{table}
\end{document}